\newif\ifcamera
\camerafalse

\ifcamera
\documentclass[twoside,leqno]{article}
\usepackage[letterpaper]{geometry}
\usepackage{ltexpprt}
\else
\documentclass[11pt]{article}
\usepackage[letterpaper, portrait, left=1in, right=1in, top=0.95in, bottom=0.95in]{geometry}
\usepackage{amsthm}
\newtheorem{theorem}{Theorem}

\newtheorem{proposition}[theorem]{Proposition}

\newtheorem{lemma}[theorem]{Lemma}
\fi 

\RequirePackage{etex}
\usepackage[algo2e, ruled,vlined, linesnumbered]{algorithm2e}
\usepackage{hyperref}
\usepackage{wrapfig}
\usepackage[dvipsnames]{xcolor}
\usepackage[utf8]{inputenc}
\usepackage[english]{babel}
\usepackage{amsmath,amsfonts}
\usepackage{amssymb}
\usepackage{autonum}
\usepackage{cleveref}
\usepackage{soul}
\usepackage{dsfont}
\usepackage{makecell}
\usepackage{multirow}
\usepackage{nicefrac}
\usepackage{bbm}
\usepackage{lineno}
\usepackage{wrapfig}
\usepackage[shortlabels]{enumitem}
\usepackage[toc]{appendix}
\usepackage{graphicx}
\usepackage[showonlyrefs]{mathtools}
\usepackage{moresize}
\usepackage{todonotes}
\usepackage[nocompress]{cite}

\newcommand{\ld}{\left}
\newcommand{\rd}{\right}
\newcommand{\etal}{\textit{et al.}}
\newcommand{\PropEstimator}{\textsc{PropEstimator}}
\newcommand{\NAPropEstimator}{\textsc{NoAdvicePropEstimator}}
\newcommand{\unisizeestimator}{\textsc{SetSizeEstimator}}
\newcommand{\propbucketsampler}{\textsc{PropBucketSampler}}
\newcommand{\unifbucketsampler}{\textsc{UnifBucketSampler}}
\newcommand{\bernoulliestimation}{\textsc{BernoulliEstimator}}
\newcommand{\harmonicestimator}{\textsc{HarmonicEstimator}}
\newcommand{\HybridEstimator}{\textsc{HybridEstimator}}
\newcommand{\NAHybridEstimator}{\textsc{NoAdviceHybridEstimator}}
\newcommand{\NACouponCollector}{\textsc{NoAdviceCouponCollector}}

\DeclareMathOperator*{\argmax}{arg\,max}

\newcommand{\punif}{P_{unif}}
\newcommand{\pprop}{P_{prop}}
\newcommand{\Bb}{\mathcal{B}}
\newcommand{\Ee}{\mathcal{E}}
\newcommand{\Rr}{\mathcal{R}}
\newcommand{\Ll}{\mathcal{L}}
\newcommand{\Dd}{\mathcal{D}}

\newcommand{\onepmepsilon}{(1\pm\varepsilon)}
\let\oldnl\nl 
\newcommand{\nonl}{\renewcommand{\nl}{\let\nl\oldnl}}

\ifcamera
\crefname{@theorem}{theorem}{theorems}
\fi

\date{}
\title{Better Sum Estimation via Weighted Sampling}
\author{
	Lorenzo Beretta \thanks{
	\begin{wrapfigure}{l}{0.055\textwidth}
\vspace{-6 mm}
\includegraphics[height=0.055\textwidth, width=0.055\textwidth]{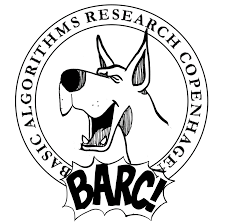}
\includegraphics[height=0.03\textwidth, width=0.055\textwidth]{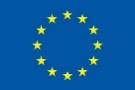}
\end{wrapfigure}
	Lorenzo Beretta and Jakub Tětek belong to Basic Algorithms Research Copenhagen (BARC), University of Copenhagen. BARC is supported by the VILLUM Foundation grant 16582.
	Jakub Tětek received funding from the Bakala Foundation.
Lorenzo Beretta receives funding from the European Union's Horizon 2020 research and innovation program under the Marie Skłodowska-Curie grant agreement No.~801199.}
\\
	\texttt{beretta@di.ku.dk}\\
	BARC, Univ. of Copenhagen
	\and
	Jakub Tětek \footnotemark[1] \\
	\texttt{j.tetek@gmail.com}\\
	BARC, Univ. of Copenhagen
}

\begin{document}

\maketitle
\ifcamera
\fancyfoot[R]{\scriptsize{Copyright \textcopyright\ 2022 by SIAM\\
Unauthorized reproduction of this article is prohibited}}
\fi

\begin{abstract} 
Given a large set $U$ where each item $a\in U$ has weight $w(a)$, we want to estimate the total weight $W=\sum_{a\in U} w(a)$ to within factor of $1\pm\varepsilon$ with some constant probability $>1/2$. Since $n=|U|$ is large, we want to do this without looking at the entire set $U$. 
In the traditional setting in which we are allowed to sample elements from $U$ uniformly, sampling $\Omega(n)$ items is necessary to provide any non-trivial guarantee on the estimate. 
Therefore, we investigate this problem in different settings: in the \emph{proportional} setting we can sample items with probabilities proportional to their weights, and in the \emph{hybrid} setting we can sample both proportionally and uniformly. These settings have applications, for example, in sublinear-time algorithms and distribution testing.

Sum estimation in the proportional and hybrid setting has been considered before by Motwani, Panigrahy, and  Xu [ICALP, 2007]. In their paper, they give both upper and lower bounds in terms of $n$. Their bounds are near-matching in terms of $n$, but not in terms of $\varepsilon$. In this paper, we improve both their upper and lower bounds. 
Our bounds are matching up to constant factors in both settings, in terms of both $n$ and $\varepsilon$. No lower bounds with dependency on $\varepsilon$ were known previously. In the proportional setting, we improve their $\tilde{O}(\sqrt{n}/\varepsilon^{7/2})$ algorithm to $O(\sqrt{n}/\varepsilon)$. In the hybrid setting, we improve $\tilde{O}(\sqrt[3]{n}/ \varepsilon^{9/2})$ to $O(\sqrt[3]{n}/\varepsilon^{4/3})$. Our algorithms are also significantly simpler and do not have large constant factors.

We then investigate the previously unexplored scenario in which $n$ is not known to the algorithm. In this case, we obtain a $O(\sqrt{n}/\varepsilon + \log n / \varepsilon^2)$ algorithm for the proportional setting, and a $O(\sqrt{n}/\varepsilon)$ algorithm for the hybrid setting. This means that in the proportional setting, we may remove the need for advice without greatly increasing the complexity of the problem, while there is a major difference in the hybrid setting. We prove that this difference in the hybrid setting is necessary, by showing a matching lower bound.

Our algorithms have applications in the area of sublinear-time graph algorithms. Consider a large graph $G=(V, E)$ and the task of $\onepmepsilon$-approximating $|E|$.
We consider the (standard) settings where we can sample uniformly from $E$ or from both $E$ and $V$. This relates to sum estimation as follows: we set $U = V$ and the weights to be equal to the degrees. Uniform sampling then corresponds to sampling vertices uniformly. Proportional sampling can be simulated by taking a random edge and picking one of its endpoints at random.
If we can only sample uniformly from $E$, then our results immediately give a $O(\sqrt{|V|} / \varepsilon)$ algorithm. 
When we may sample both from $E$ and $V$, our results imply an algorithm with complexity $O(\sqrt[3]{|V|}/\varepsilon^{4/3})$. 
Surprisingly, one of our subroutines provides an $\onepmepsilon$-approximation of $|E|$ using  $\tilde{O}(d/\varepsilon^2)$ expected samples, where $d$ is the average degree, under the mild assumption that at least a constant fraction of vertices are non-isolated. This subroutine works in the setting where we can sample uniformly from both $V$ and $E$. We find this remarkable since it is $O(1/\varepsilon^2)$ for sparse graphs.
\end{abstract}

\section{Introduction.} \label{sec:intro}
Suppose we have a large set $U$, a weight function $w:U \rightarrow [0,\infty)$ and we want to compute a $\onepmepsilon$-approximation of the sum of all weights $W = \sum_{a\in U} w(a)$. Since $n=|U|$ is very large, we want to estimate $W$ by sampling as few elements as possible. 
In the traditional setting in which we are allowed to sample elements from $U$ uniformly, sampling $o(n)$ items cannot provide any non-trivial guarantee on the approximation as we may miss an element with a very large weight.
This led Motwani, Panigrahy, and  Xu \cite{Motwani2007} to study this problem when we are allowed to sample elements proportionally to their weights (i.e., sample $a$ with probability $w(a)/W$). In particular, they studied two settings: the \emph{proportional} setting, where we can sample items proportionally to their weights, and the \emph{hybrid} setting where both proportional and uniform sampling is possible. In this paper, we revisit these two settings and get both improved lower and upper bounds. We also extend the results to more general settings and show how our techniques imply new results for counting edges in sublinear time.
 
Motwani et al. \cite{Motwani2007} give upper and lower bounds for both proportional and hybrid settings. Their bounds are matching up to polylogarithmic factors in terms of $n$, but not in terms of $\varepsilon$. In this paper, we improve both their upper and lower bounds. 
Our bounds are matching up to \emph{constant} factors in both settings, in terms of both $n$ and $\varepsilon$. No lower bounds with dependency on $\varepsilon$ were known previously.

In the proportional setting, we improve their $\tilde{O}(\sqrt{n}/\varepsilon^{7/2})$ algorithm to $O(\sqrt{n}/\varepsilon)$. In the hybrid setting, we improve $\tilde{O}(\sqrt[3]{n}/ \varepsilon^{9/2})$ to $O(\sqrt[3]{n}/\varepsilon^{4/3})$. Our algorithms are also significantly simpler. In the same paper, Motwani \etal{} \cite{Motwani2007} write: ``to efficiently derive the sum from [proportional] samples does not seem straightforward". We would like to disagree and give a formula that outputs an estimate of $W$ from a proportional sample. This formula is not only optimal in terms of sample complexity but also very simple. Our other algorithms, although not as simple, do not have large hidden constants and we believe they are both practical and less involved than their predecessors.

In their work, Motwani \etal{} \cite{Motwani2007} always assume to know the size of the universe $n = |U|$. We extend these results to the case of unknown $n$.
In this case, we obtain a $O(\sqrt{n}/\varepsilon + \log n/\epsilon^2)$ algorithm for the proportional setting, and a $O(\sqrt{n}/\varepsilon)$ algorithm for the hybrid setting. This means that in the proportional setting we may remove the need for advice without significantly impacting the complexity of the problem, while there is a major difference in the hybrid setting. We prove that this difference in the hybrid setting is necessary, by showing a matching lower bound.

We give lower bounds for all our estimation problems, both for proportional and hybrid settings, and when $n$ is either known or unknown. This is the most technically challenging part of the paper.
We prove lower bounds for $n$ known and unknown as well as proportional and hybrid settings; all our lower and upper bounds are matching up to a constant factor. See \Cref{tab:summary_results} for a summary of our results.

Our algorithms have particularly interesting applications in the area of sublinear-time graph algorithms, which are explained in detail in \Cref{sec:rel-work-and-applications}.

The paper is structured as follows. In what is left of \Cref{sec:intro} we explain applications and related work, give an overview of the techniques employed, and provide the reader with formal definitions of problems and notations.  \Cref{sec:weighted_sampling} contains our algorithms for proportional setting, while \Cref{sec:hybrid_sampling} contains algorithms for hybrid setting. In \Cref{sec:lower_bounds_section} we prove all our lower bounds. In \Cref{sec:approximatelycountingedges} we show how to apply our algorithms to the problem of counting the number of edges in a graph in sublinear time. In \Cref{sec:open_problems} we raise several open problems.

\def\arraystretch{1.9}
\begin{table}[ht]
\resizebox{\textwidth}{!}{%
\begin{tabular}{|l|ll|ll|}
\hline 
\multirow{2}{*}{\makecell[l]{Advice to\\ the algorithm}}& \multicolumn{2}{|c|}{This paper} & \multicolumn{2}{|c|}{Motwani, Panigrahy and  Xu \cite{Motwani2007}}   \\
\cline{2-5}
                                    &  Proportional & Hybrid                                                             & Proportional      & Hybrid \\ \hline \hline
 \makecell[l]{\\$n$ known\\\phantom{ }}                          & 
 $\Theta(\sqrt{n}/\varepsilon)$  &  \makecell[l]{$O(\min(\sqrt[3]{n}/\varepsilon^{4/3},n\log n))$ \\ $\Omega(\min(\sqrt[3]{n} / \varepsilon^{4/3}, n)$}  & \makecell[l]{$\Tilde{O}(\sqrt{n}/\varepsilon^{7/2})$\\ $\Omega(\sqrt{n})$} & \makecell[l]{$\Tilde{O}(\sqrt[3]{n}/\varepsilon^{9/2}),\, O(\sqrt{n}/\varepsilon^2)$ \\  $\Omega(\sqrt[3]{n})$} \\ \hline
 \makecell[l]{Known $\tilde{n} \geq n$}    &  \makecell[l]{$O(\sqrt{\tilde{n}}/\varepsilon)$ \\ $\Omega(\sqrt{n}/\varepsilon)$} & \makecell[l]{ $O(\min(\sqrt{n}/\varepsilon,n\log n))$ \\ $\Omega(\min(\sqrt{n}/\varepsilon, n))$}                                                     &  &  \\ \hline
 \makecell[l]{\\No advice\\\phantom{ }}&  \makecell[l]{$O(\sqrt{n}/\varepsilon + \log n/\epsilon^2)$ \\ $\Omega(\sqrt{n}/\varepsilon)$}                         &  \makecell[l]{ $O(\min(\sqrt{n}/\varepsilon,n\log n))$ \\ $\Omega(\min(\sqrt{n}/\varepsilon, n))$}                                                      &  &  \\ \hline
\end{tabular}
}
\caption{Results of this paper.} 
\label{tab:summary_results}
\end{table}




\subsection{Related Work and Applications.} \label{sec:rel-work-and-applications}
In this section, we show that our algorithms can be applied to get sublinear-time graph algorithms and distribution testing and discuss how this relates to previous work in these areas. We also discuss here how the widely used Metropolis-Hastings algorithm in fact implements the proportional sampling.
We suggest that this could be an application domain worth investigating. 

\subsubsection*{Counting edges in sublinear time.}
When a graph $G=(V, E)$ is very large, we may want to approximately solve certain tasks without looking at the entire $G$, thus having a time complexity that is sublinear in the size of $G$.
In particular, estimating global properties of $G$ such $|V|$ or $|E|$ in this setting is an important problem and has been studied in both theoretical and applied communities \cite{Feige2004,Goldreich2006,Seshadhri2015,Eden2017,Tetek_Thorup,Katzir2011,Dasgupta2014}. 
Since the algorithm does not have the time to pre-process (or even see) the whole graph, it is important to specify how we access $G$. Several models are employed in the literature.
The models differ from each other for the set of queries that the algorithm is allowed to perform. A \emph{random vertex query} returns a random vertex, a \emph{random edge query} returns a random edge, a \emph{pair query} takes two vertices $u, v$ as arguments and returns $(u, v) \in E$, a \emph{neighbourhood query} takes a vertex $v$ and an index $i$ as arguments and returns the $i$-th neighbour of $v$ (or says that $d(v) < i$), a \emph{degree query}, given a vertex $v$, returns its degree $deg(v)$. We parameterize the complexities with an approximation parameter $\varepsilon$, $n=|V|$ and $m=|E|$.

The problem of estimating the number of edges in a graph in sublinear time has been first considered by Feige \cite{Feige2004}. Their algorithm works in the model where only random vertex queries  and degree queries are allowed and achieves a $(2+\varepsilon)$-approximation algorithm using $\tilde{O}(\frac{n}{\varepsilon \sqrt{m}})$ time and queries. Their algorithm does not use neighborhood queries and the authors showed that without neighborhood queries, $2-\varepsilon$ approximation requires a linear number of queries. Goldreich and Ron \cite{Goldreich2006} broke the barrier of factor $2$ by using neighborhood queries. Indeed, they showed a $(1+\varepsilon)$-approximation with time and query complexity of $\tilde{O}(\frac{n}{\varepsilon^{9/2} \sqrt{m}})$. Currently, the best known algorithm is the one by Eden, Ron and Seshadhri \cite{Eden2017} and has complexity $\smash{\tilde{O}}(\frac{n}{\varepsilon^{2} \sqrt{m}})$. If pair queries  are allowed, the algorithm of Tětek and Thorup \cite{Tetek_Thorup} has complexity $\smash{\tilde{O}}(\frac{n}{\varepsilon \sqrt{m}}+ \frac{1}{\varepsilon^4})$; in the same paper the authors showed a lower bound which is near-matching for $\varepsilon \geq m^{1/6}/ n^{1/3}$ as well as an algorithm with complexity $\smash{\tilde{O}}(\frac{n}{\varepsilon \sqrt{m}}+ \frac{1}{\varepsilon^2})$ in what they call the hash-ordered access model. They also show an algorithm in the more standard setting with random vertex and neighborhood queries, that runs in time $\tilde{O}(\sqrt{n}/\varepsilon)$. This is the same complexity (up to a $\log^{O(1)} n$ factor) that we achieve in the setting with random edge queries, as we discuss below. Our techniques are, however, completely different and share no similarity with the techniques used in \cite{Tetek_Thorup}.

Our algorithms can be applied to solve the edge counting problem when either (i) random edge queries only are allowed, or (ii) both random edge and random vertex queries are allowed. While edge counting has not been explicitly considered in these settings before, these settings are established and have been used in several papers \cite{Assadi2018,Aliakbarpour2018,Fichtenberger2020,Biswas2021}. We instantiate our algorithm for this graph problem setting $U = V$ and $w(v) = deg(v)$ for each $v \in V$. Uniform sampling then corresponds to sampling vertices uniformly. Proportional sampling can be simulated by taking a random edge and picking one of its endpoints at random. Degree query allows us to get the weights of sampled vertices. Since these settings have not been explicitly studied before, we compare our results with what follows directly from the known literature.

If we can only sample uniformly from $E$, the algorithm by Motwani et al.\ implies an algorithm for this problem that has complexity $\Tilde{O}(\sqrt{n}/\varepsilon^{7/2})$.
Using an algorithm from \cite{Eden2017} and standard simulation of random vertex queries using random edge queries, one would get time $\tilde{O}(\sqrt{m}/\varepsilon^2 + \frac{m}{\varepsilon \sqrt{n'}})$ for $n'$ being the number of non-isolated\footnote{A vertex is isolated if it has degree zero.} vertices\footnote{One may simulate uniform sampling from the set of non-isolated vertices at multiplicative overhead of $O(m/n')$ by sampling proportionally and using rejection sampling. We may then use set size estimation by bithday paradox in time $O(\sqrt{n'}/\epsilon)$ to learn $n'$ and the algorithm of \cite{Eden2017}.}. Our results immediately give a $O(\sqrt{n} / \varepsilon)$ algorithm, or $\tilde{O}(\sqrt{n} / \varepsilon + 1/\varepsilon^2)$ when $n$ is not known. 

When we may sample both from $E$ and $V$, the algorithms by Motwani et al.\ imply algorithms with sample complexities of $\Tilde{O}(\sqrt[3]{n}/\varepsilon^{9/2})$ and $O(\sqrt{n}/\varepsilon^2)$. We may also use the algorithm of \cite{Eden2017} that relies on random vertex query only. This has complexity from $O(\frac{n}{\varepsilon^2\sqrt{m}})$. Our results imply an algorithm with complexity $O(\sqrt[3]{n}/\varepsilon^{4/3})$. 
Surprisingly, one of our subroutines provides an $\onepmepsilon$-approximation of $|E|$ using  $\tilde{O}(d/\varepsilon^2)$ expected samples, where $d$ is the average degree, under the mild assumption that at least a constant fraction of vertices are non-isolated (in fact, we prove a more complicated complexity which depends on the fraction of vertices that are isolated). This subroutine works in the setting where we can sample uniformly from both $V$ and $E$. We find this remarkable since it is $O(1/\varepsilon^2)$ for sparse graphs.


\subsubsection*{Distribution testing.}
Consider a model in which we are allowed to sample from a distribution $\Dd$ on $U$, and when we do, we receive both $a \in U$ and $P_\Dd(a)$. This model is stronger
than the one we considered throughout the paper. In fact, it is a special case of our model obtained by fixing $W=1$. Similarly, we may consider such stronger variant of the hybrid setting.
This model has received attention both in statistics and in theoretical computer science literature. Most notably, Horvitz and Thompson \cite{horvitz_1952} in 1952 showed\ how to estimate the sum $\sum_{a\in U} v(a)$, for value function $v: U \rightarrow \mathbb{R}$ when having access to a sample $a$ from $\mathcal{D}$ and $P_\mathcal{D}(a)$. 
More recently, many distribution testing problems have been considered in this and similar models.

First, Canonne and Rubinfeld \cite{Canonne2014} considered the model where one can (i) sample $a\in U$ according to $\Dd$; (ii) query an oracle that, given $a \in U$, returns $P_\Dd(a)$. 
This allows us to both get the probability of a sampled item, but also the probability of any other item. For every permutation-invariant problem, any sublinear-time algorithm in this model may be transformed so that it only queries the oracle on previously sampled items or on items chosen uniformly at random \footnote{Since we consider permutation invariant problems, we may randomly permute the elements. Whenever we query an element that has not yet been sampled, we may assume that it is sampled uniformly from the set of not-yet-sampled elements. This can be simulated by sampling elements until getting a not-yet-seen item (sublinearity ensures that the step adds multiplicative $O(1)$ overhead).}. This setting is stronger than our hybrid setting. However, the two models become equivalent if we know $W$.

Second, Onak and Sun \cite{Onak2018} considered exactly the model described above, where one can sample from a distribution $\Dd$, and both $a\in U$ and $P_\Dd(a)$ are returned. This setting is again stronger than our proportional setting, and it becomes equivalent once we know $W$.

In both of these papers, the authors solved several distribution testing problems such as: uniformity testing, identity testing, and distance to a known distribution, in the respective models. Canonne and Rubinfeld \cite{Canonne2014} also considered the problem of providing an additive approximation of entropy.
Both \cite{Canonne2014} and \cite{Onak2018} proved several lower bounds, showing that many of their algorithms are optimal up to constant factors. These lower bounds also imply a separation between the two models for the problems mentioned above.

In both of the papers, the authors make a point that many of their algorithms (all the ones we mentioned) are robust with respect to noise of multiplicative error $(1\pm \varepsilon/2)$ in the answers of the probability oracles.
Our algorithms can then serve as a reduction from the weaker models we consider in this paper to these stronger models where we know the probabilities and not just the weights. The reason is that we can get an approximation of $W$, meaning that we may then approximately implement the above-described models. Since the mentioned algorithms \cite{Canonne2014,Onak2018} are robust,
a $(1\pm\varepsilon/2)$-approximation of $W$ is sufficient to simulate them.

\subsubsection*{Proportional sampling in practice: Metropolis-Hastings.}
Proportional sampling is often implemented in practice using the Metropolis-Hasting algorithm. This algorithm is widely used in statistics and statistical physics, but can also be used to sample combinatorial objects. It can be used to sample from large sets which have complicated structures that make it difficult to use other sampling methods. Just like our algorithms, it is suitable when the set is too large to be stored explicitly, making it impossible to pre-process it for efficient sampling. One of the main appeals of Metropolis-Hasting is that it does not require one to know the exact sampling probabilities, but it is sufficient to know the items' weights like in the proportional setting described in this paper. We believe that our algorithm can find practical applications in combination with Metropolis-Hasting as Metropolis-Hasting performs proportional sampling with weight function that would be usually known in practice.

\subsection{Overview of employed techniques.}
Here we provide a summary of the techniques employed throughout the paper. We denote with $\punif(\cdot)$ and $\pprop(\cdot)$ the probabilities computed according to uniform and proportional sampling, respectively. Although often not specified for brevity, all guarantees on the approximation factors of estimates in this section are meant to hold with probability $2/3$.


\subsubsection{Proportional setting with advice \texorpdfstring{$\tilde{n} \geq n$}{ntgn}.}
Consider sampling two elements $a_1,a_2\in U$ proportionally and define $Y_{12} = 1/w(a_1)$ if $a_1 = a_2$, and $Y_{12}=0$ otherwise. It is easy to show that $Y_{12}$ is an unbiased estimator of $1/W$. We could perform this experiment many times and take the average. This would give a good approximation to $1/W$ and taking the inverse value, we would get a good estimate on $W$. Unfortunately, we would need $\Theta(n/\varepsilon^2)$ repetitions in order to succeed with a constant probability.
We can fix this as follows: we take $m$ samples $a_1, \cdots, a_m$ and consider one estimator $Y_{ij}$ for each pair of samples $a_i, a_j$ for $i \neq j$. This allows us to get $\binom{m}{2}$ estimators from $m$ samples. We show that estimators $Y_{ij}$ are uncorrelated. This reduces the needed number of samples from $O(n/\varepsilon^2)$ to $O(\sqrt{n}/\varepsilon)$.

We now describe the estimator formally. Let $S = \{a_1, \cdots, a_m\}$ be the set of sampled items, and for each $s \in S$ define $c_s$ to be the number of times item $s$ is sampled. Then,
\[
\hat{W} = \binom{m}{2} \cdot \left(\sum_{s \in S} \frac{\binom{c_s}{2}}{w(s)}\right)^{-1}
\]
is a $\onepmepsilon$-approximation of $W$ with probability $2/3$ for $m= \Theta(\sqrt{n}/\varepsilon)$. If, instead of knowing $n$ exactly, we know some $\tilde{n} \geq n$ we can achieve the same guarantees by taking $\Theta(\sqrt{\tilde{n}}/\varepsilon)$ samples.


\subsubsection{Proportional setting, unknown \texorpdfstring{$n$}{n}.} \label{subsetction_prop_unknownn}
We partition $U$ into buckets $B_i = \{a\in U \,|\, w(a) \in [2^i, 2^{i+1})\}$.
We choose some $b\in \mathbb{Z}$ and compute two estimates: $\hat{W}_b$ approximates $W_b = \sum_{a\in B_b} w(a)$, and $\hat{P}_b$ approximates $P_b =\pprop(a \in B_b) = W_b / W$. We then return $\hat{W}_b/ \hat{P}_b$, as an estimate of $W$. If both estimates are accurate, then the returned value is accurate. To choose $b$, we sample $a_1$ and $a_2$ proportionally and define $b$ so that $\max\{a_1, a_2\} \in [2^b, 2^{b+1})$. We prove that, surprisingly\footnote{Notice that, if we just define $b=a_1$, then we have $E[1/P_b] = n$ in the worst case. Therefore, taking the maximum of two samples entails an exponential advantage.}, $E[1/P_b] = O(\log n)$. Computing $\hat{P}_b$ is simply a matter of estimating the fraction of proportional samples falling into $B_b$. This can be done using $O(1/(P_b \varepsilon^2))$ samples, where $\varepsilon$ is the approximation parameter. Therefore the expected complexity of computing $\hat{P}_b$ is $O(\log n/ \varepsilon^2)$. Computing $\hat{W}_b$ is more involved. First, we design two subroutines to sample from $B_b$, one for proportional sampling and one for uniform. These subroutines work by sampling $a\in U$ until we get $a \in B_b$; the subroutine for uniform sampling then uses rejection sampling. These subroutines take, in expectation, $O(1/P_b)$ samples to output one sample from $B_b$. 
Since we can now sample uniformly from $B_b$, we sample items uniformly form $B_b$ until we find the first repeated item and use the stopping time to infer $|B_b|$ up to a constant factor. In expectation, we use $O(\sqrt{|B_b|})$ uniform samples and compute $\tilde{n}_b$, such that $|B_b| \leq \tilde{n}_b \leq O(|B_b|)$.
Since we can also sample proportionally from $B_b$, we may use the algorithm for proportional setting with advice $\tilde{n}_b$ to estimate $W_b$. 
This yields a total sample complexity of $O(\sqrt{n} / \varepsilon + \log n / \varepsilon^2)$.

\subsubsection{Hybrid setting.} We now present the techniques used in the hybrid setting. Before giving a sketch of the main algorithms, we introduce two subroutines.

\paragraph{Coupon-collector-based algorithm.}
In the hybrid setting, we can sample elements uniformly. If we know $n=|U|$ a well-known result under the name of ``coupon collector problem'' shows that we can retrieve with high probability all $n$ elements by performing $\Theta(n \log n)$ uniform samples. We extend this result to the case of unknown $n$. We maintain a set of retrieved elements $S \subseteq U$ and keep on sampling uniformly and adding new elements to $S$ until we perform $\Theta(|S|\log |S|)$ samples in a row without updating $S$. It turns out that this procedure retrieves the whole set $U$ with probability $2/3$ and expected complexity $O(n \log n)$. According to our lower bounds, this algorithm is near-optimal when $\varepsilon = O(1/\sqrt{n})$. Therefore, we can focus on studying hybrid sampling for larger values of $\varepsilon$.

\paragraph{Harmonic mean and estimating \texorpdfstring{$W/n$}{W/n} with advice \texorpdfstring{$\tilde{\theta} \geq W/n$}{theta >= W/n} in hybrid setting.}
If we sample $a \in U$ proportionally, then $1/w(a)$ is an unbiased estimator of $n/W$. Unfortunately, we may have very small values of $w(a)$, which can make the variance of this estimator arbitrarily large. To reduce variance, we can set a threshold $\phi$ and define an estimator $Y_\phi$ as $1/w(a)$ if $w(a) \geq \phi$, and $0$ otherwise. 
Set $p=\punif(w(a) \geq \phi)$, then we have $E[Y_\phi]=p n/ W$ and $Var(Y_\phi) \leq p n / (\phi W)$.
If we define $\Bar{Y}_\phi$ as the average of $\tilde{\theta}/(p\phi \varepsilon^2)$ copies of $Y_\phi$, we have $Var(\Bar{Y}_\phi) \leq (\varepsilon E[Y_\phi])^2$. With such a small variance, $\Bar{Y}_\phi$ is a good estimate of $p n / W$.
Estimating $p$ is simply a matter of estimating the fraction of uniform samples having $w(\bullet) \geq \phi$. This can be done taking $O(1/(p\varepsilon^2))$ uniform samples. Once we have estimates of both $p$ and $p n / W$ we return their ratio as an estimate of $W/n$. We employed in total $O((1 +\frac{\tilde{\theta}}{\phi}) / (p \, \varepsilon^2))$ proportional and uniform samples.

Perhaps surprisingly, this corresponds to taking the harmonic mean of the samples with weight at least $\phi$ and adjusting this estimate for the weight of the items with weight $< \phi$.

\paragraph{Hybrid setting, known $n$.}
We now sketch the algorithm for sum estimation in the hybrid setting with known $n$. We combine our formula to estimate $W$ in the proportional setting and the above algorithm to estimate $W/n$ in hybrid setting to obtain an algorithm that estimates $W$ using $O(\sqrt[3]{n}/\varepsilon^{4/3})$ uniform and proportional samples. Define $U_{\geq \theta}=\{a\in U \,|\, w(a) \geq \theta\}$. We find a $\theta$ such that $\punif(a \in U_{\geq \theta}) \approx n^{-1/3}\varepsilon^{-2/3}$, this can be done taking enough uniform samples and picking the empirical $(1-n^{-1/3}\varepsilon^{-2/3})$-quantile
\footnote{We may assume that $\epsilon \geq 8/\sqrt{n}$ by running the coupon-collector-based algorithm when $\epsilon < 8/\sqrt{n}$. Under this assumption, it holds $n^{-1/3}\varepsilon^{-2/3} \in [0, 1]$ and taking the $(1-n^{-1/3}\varepsilon^{-2/3})$-quantile then is meaningful.}
.
We define $p=\pprop(a \in U_{\geq \theta}) \geq \punif(a \in U_{\geq \theta}) \approx n^{-1/3}\varepsilon^{-2/3}$. We compute an estimate $\hat{p}$ of $p$ counting the fraction of proportional samples falling in $U_{\geq \theta}$.
Now we have two cases. 
If $\hat{p} \geq 1/2$, we can simulate sampling from the proportional distribution on $U_{\geq \theta}$ with $O(1)$ overhead by sampling proportionally until we get an element of $U_{\geq \theta}$. Then, we use the algorithm for sum estimation under proportional sampling restricted to elements in $U_{\geq \theta}$ to estimate $\sum_{a \in U_{\geq \theta}} w(a) = p W$. Dividing by the estimate $\hat{p} \approx p$, we get an estimate for $W$.
Else, $\hat{p} < 1/2$, and thus $p \leq 2/3$ (assuming $\hat{p}$ is a good enough estimate of $p$). We then have $\theta n \geq \sum_{a \not\in U_{\geq \theta}} w(a) = (1-p) W \geq W / 3$. This allows us to use the harmonic-mean-based algorithm to estimate $W/n$ with $\phi = \theta$ and $\tilde{\theta} = 3\theta$.
In both cases we manage to provide an estimate of $W$ using $O(\sqrt[3]{n} / \varepsilon^{4/3})$ samples.

\paragraph{Hybrid setting, unknown $n$.}
We now sketch our technique for sum estimation in the hybrid setting with unknown $n$.
First, we sample items uniformly until we find the first repeated item and use the stopping time to infer the total number of items $n$. In expectation, we use $O(\sqrt{n})$ uniform samples and compute $\tilde{n}$, such that $n \leq \tilde{n}$ and $E[\tilde{n}] = n$. Now, we can use our algorithm for sum estimation in the proportional setting with advice $\tilde{n}$. This uses in expectation $O(\sqrt{\tilde{n}}/\varepsilon) = O(\sqrt{n} / \varepsilon)$ samples. As we prove, this complexity is optimal up to a constant factor for $\epsilon \geq 1/\sqrt{n}$.
Again, if $\varepsilon \leq 1/\sqrt{n}$, then we use our coupon-collector-based algorithm that retrieves every element of $U$ using $O(n \log n)$ samples.

\subsubsection{Lower bounds.}
The most technically challenging part of our paper is \Cref{sec:lower_bounds_section}, where we prove lower bounds for all estimation problems we address in the first part of the paper. 
All our lower bound proofs follow a common thread. We now sketch the main ideas. First, we define two different instances of the estimation problem at hand $(U_1, w_1)$ and $(U_2, w_2)$ such that a $\onepmepsilon$-approximation of $W$ is sufficient to distinguish between them. Then, we define our \emph{hard} instance as a mixture of the two: we take $(U_1, w_1)$ with probability $1/2$ and $(U_2, w_2)$ otherwise. We denote these events by $\Ee_1$ and $\Ee_2$ respectively. 
Second, we show that a Bayes classifier cannot distinguish between the two cases with probability $2/3$ using too few samples; since Bayes classifiers are risk-optimal\footnote{Risk of a classifier refers to the misclassification probability under some fixed distribution. For the Bayes classifier, we implicitly assume that this distribution is the same as the prior.} this implies that no classifier can have misclassification probability less than $1/3$ while using the same number of samples. 
To show that a Bayes classifier has a certain misclassification probability, we study the posterior distribution, conditioned on the samples it has seen. Let the multiset $S$ represent the outcome of the samples. Since the prior is uniform, applying the Bayes theorem gives
\[
\frac{P\ld(\Ee_1 \,|\, S\rd)}{P\ld(\Ee_2 \,|\, S\rd)} = \frac{P\ld( S \,|\,\Ee_1 \rd)}{P\ld(S \,|\,\Ee_2\rd)}.
\]
We denote this likelihood ratio by $\Rr(S)$. We show that whenever $|S|$ is (asymptotically) too small, we have $\Rr(S) \approx 1$ with probability close to $1$. When $\Rr(S) \approx 1$, the posterior distribution is very close to uniform. This entails misclassification probability close to $1/2$.
If $X$ and $Y$ are some random variables sufficient to reconstruct $S$ (that is, there exists an algorithm that, given $(X,Y)$, generates $S'$ with the same distribution as $S$), we can define
\[
\Rr(X) = \frac{P\ld(X\,|\, \Ee_1\rd) }{P\ld(X \,|\, \Ee_2\rd)}\quad \text{ and } \quad \Rr(Y\,|\,X) = \frac{P\ld(Y \,|\, X, \Ee_1\rd)}{P\ld(Y \,|\, X, \Ee_2\rd)}
\]
and we have, thanks to the Bayes theorem, $\Rr(S) = \Rr(X) \cdot \Rr(Y\,|\,X)$. In this way, we can break the problem of proving $\Rr(S) \approx 1$ into proving $\Rr(X) \approx 1$ and $\Rr(Y\,|\,X) \approx 1$. This allows us to reduce all our lower bounds to prove concentration of likelihood ratios for two basic problems: (1) distinguishing between two sets of size $n$ and $(1-\varepsilon)n$ by uniform sampling and (2) distinguishing between two sequences of i.i.d. random variables from $Bern(p)$ and $Bern(p - \varepsilon)$. 

We now sketch how we bound the likelihood ratio $\mathcal{R}(S)$ for problem (1), the same technique applies to problem (2).
In problem (1), we call $\Ee_1$ the event $|U|=n$ and $\Ee_2$ the event $|U|=(1-\varepsilon) n$, and we set $P(\Ee_1)=P(\Ee_2)=1/2$. 
First, we notice that $\Rr(S)$ depends only on the number $\ell(S)$ of distinct elements in $S$. Then, we prove three facts. First, $\ell(S)\,|\,\Ee_i$ is concentrated around $E[\ell(S)\,|\,\Ee_i]$. Second, $E[\ell(S) \,|\, \Ee_1] \approx E[\ell(S) \,|\, \Ee_2]$. Third, for small deviations of $\ell(S)$, we have small deviations of $\Rr(\ell(S))$. These three facts are sufficient to conclude that, with probability close to $1$, $R(\ell(S))$ lies in a very narrow interval; further computations show that $1$ lies in that interval, hence $\Rr(S) \approx 1$ with probability close to $1$.

\subsection{Preliminaries.} \label{sec:preliminaries}

\paragraph{Problem definition.} We now give a formal definition of the two settings that we consider.
Let us have a set $U$ of cardinality $n$ and a weight function $w: U \rightarrow [0, \infty)$. We denote by $W$ the sum $\sum_{a\in U} w(a)$.
The following operations are allowed in the proportional sampling setting: (1) proportionally sample an item, this returns $(a, w(a))$ with probability $w(a)/W$; (2) given two items $a,a'$, check whether $a=a'$. This is the only way we can interact with the items. In the hybrid setting, we may in addition (3) sample an item uniformly (that is, return $(a,w(a))$ for any $a \in U$ with probability $1/n$).
In both settings, we want to compute an estimate $\hat{W}$ of $W$ such that $(1-\varepsilon)W \leq \hat{W} \leq (1+\varepsilon)W$ with probability $2/3$. 

\paragraph{Notation.}
When $(1-\varepsilon)W \leq \hat{W} \leq (1+\varepsilon)W$ holds, we say that such $\hat{W}$ is a $\onepmepsilon$-approximation of $W$.
Some of our subroutines require ``advice" in the form of a constant factor approximation of some value. For sake of consistency, we denote this constant factor approximation of $\bullet$ by $\tilde{\bullet}$. Similarly, if we want to estimate some value $\bullet$, we use $\hat{\bullet}$ to denote the estimate. Let us have some predicate $\phi$ that evaluates true on some subset of $U$ and false on the rest. We denote by $\punif(\phi(a))$ and $\pprop(\phi(a))$ the probability of $\phi$ evaluating to true for $a$ being picked uniformly and proportionally, respectively.

We use $\tilde{O}$ with the slightly non-standard meaning of $f(n) \in \tilde{O}(g(n))$ being equivalent to $f(n) \in g(n) \log^{O(1)} n$, rather than $f(n) \in g(n) \log^{O(1)} g(n)$.
We state all our results (both upper and lower bound) for some constant success probability $> 1/2$. These probabilities can be amplified to any other constants without increasing the asymptotic complexity. In pseudocode, we often say that we execute some algorithm with some failure probability. By this, we mean that one uses probability amplification to achieve that failure probability.

\paragraph{Relative bias estimation of a Bernoulli random variable.}
Let $X_1, X_2, \cdots$ be i.i.d.\ random variables distributed as $Bern(p)$. \cite{Lipton1993} gave a very simple algorithm that returns $\hat{p}$ such that, with probability at least $2/3$, $\hat{p}$ is a $\onepmepsilon$-approximation of $p$\footnote{In that paper, the authors in fact solve a more general problem. For presentation of this special case, see \cite{Watanabe2005}.}. It can be summarized as follows.
\begin{proposition}[follows from \cite{Lipton1993}] \label{claim:bernoulli_estimation}
Let $X_1, X_2, \cdots$ be i.i.d.\ random variables distributed as $Bern(p)$. There exists an algorithm that uses in expectation $O(\frac{1}{\varepsilon^2 p})$ samples and returns $\hat{p}$ such that $E[1/\hat{p}] = 1/p$ and
\[
P\ld(|\hat{p} -p| > \varepsilon p\rd) \leq \frac{1}{3}. 
\]
\end{proposition}
We call this algorithm $\bernoulliestimation(\varepsilon)$. We assume that this algorithm has access to the sequence $X_1, X_2, \cdots$; we specify these random variables when invoking the algorithm.
%

\paragraph{Probability amplification and expected values.}
Consider an estimator that gives a guarantee on the estimate $\hat{x}$ that holds with some probability (say, guarantee that a proposition $\phi(\hat{x})$ holds with probability at least $2/3$) and at the same time, we know that $E[\hat{x}] \leq y$ for some value $y$. We sometimes need to amplify the probability of the guarantee (that is, amplify the probability that $\phi(\hat{x})$ holds) but would like to retain a bound $E[\hat{x}] = O(y)$. We now argue that using the standard median trick is sufficient. Namely, we prove that 

\begin{lemma} \label{lemma:prob_amplification_expectations}
Let us have non-negative i.i.d. random variables $X_1 \cdots X_{2t-1}$ for some integer $t$, and let $X = \text{median}(X_1 \cdots X_{2t-1})$. It holds $E[X] \leq 2 E[X_1]$. 
\end{lemma}
\begin{proof}
Let $X_1', \cdots, X_{2t-1}'$ be the random variables $X_1, \cdots, X_{2t-1}$ sorted in increasing order. We then have

\begin{align}
E[X] &\leq E\left[\frac{1}{t} \cdot \sum_{i=t}^{2t-1} X_i'\right] \leq E\left[\frac{1}{t} \cdot \sum_{i=1}^{2t-1} X_i\right] \leq \frac{2t-1}{t} \cdot E[X_1]. 
\end{align}
\end{proof}

\section{Sum Estimation by Proportional Sampling.} \label{sec:weighted_sampling}
In this section, we focus on sum estimation in the proportional setting.
We design algorithms to estimate $W$ and our objective is to minimize the total number of samples taken in the worst case. We present two different algorithms that provide an $\onepmepsilon$-approximation of $W$ with probability $2/3$. The first one, $\PropEstimator$, assumes to have an upper bound on the number of elements $\tilde{n} \geq n$, and achieves sample complexity of $O(\sqrt{\tilde{n}}/\varepsilon)$. The second one, $\NAPropEstimator$, does not assume any knowledge of $n$, and produce an $\varepsilon$-estimate using $O(\sqrt{n}/\varepsilon + \log n/\varepsilon^2)$ samples in expectation.

\subsection{Algorithm with advice \texorpdfstring{$\tilde{n} \geq |U|$}{ntilde-geq-n}.}\label{sec:algo-with-advice}
Let $a_1 \dots a_m$ be $m$ items picked independently at random from $U$ with probabilities proportional to their weights. Let $S$ be the set of sampled items, and for each $s \in S$ define $c_s$ to be the number of times item $s$ is sampled. For each $i, j \in [m]^2$ define $Y_{ij}$ to be $1/w(a_i)$ if $a_i = a_j$ and $0$ otherwise. We now estimate $W$ as follows:

\medskip
\subsubsection*{Algorithm 1: {\it $\PropEstimator(\tilde{n}, \varepsilon)$}:} \addtocounter{algocf}{1}
Given a parameter $0 < \varepsilon < 1$ and advice $\tilde{n}\geq n$, perform $m = \sqrt{24\tilde{n}}/ \varepsilon+1$ samples and return the estimate 
\[
\hat{W} = \binom{m}{2} \cdot \left(\sum_{s \in S} \frac{\binom{c_s}{2}}{w(s)}\right)^{-1}.
\]
In case $c_s = 1$ for all $s \in S$ set $\hat{W}=\infty$.

\medskip \noindent
Before we prove correctness, we need the following lemma.
\begin{lemma} \label{lem:onepair}
Given pairwise distinct $i, j, k \in [m]$, we have
$E[Y_{i,j}] = 1/W$, $Var(Y_{i,j}) \leq n/W^2$, and $Cov[Y_{i,j},Y_{i,k}] = 0$
\end{lemma}
\begin{proof}\mbox{}\\*
\begin{align}
& E[Y_{ij}] = \sum_{a \in U} \frac{1}{w(a)} P(x_i = x_j = a) = \sum_{a \in U}  \frac{w(a)}{W^2} = \frac{1}{W}
\\
& Var(Y_{ij}) \leq E[Y_{ij}^2] = \sum_{a \in U} \frac{1}{w(a)^2} P(x_i = x_j = a) = \sum_{a \in U} \frac{1}{W^2} = \frac{n}{W^2} \\
\intertext{As for the covariance, it holds $Cov[Y_{i,j}, Y_{i,k}] = E[Y_{i,j} \cdot Y_{i,k}] - E[Y_{j,k}] \cdot E[Y_{i,k}]$. This is equal to $0$ as}
& E[Y_{i,j} \cdot Y_{i,k}] = \sum_{a \in U} \frac{1}{w(a)^2} P(x_i = x_j = x_k = a) = \sum_{a \in U} \frac{w(a)}{W^3} = \frac{1}{W^2} = E[Y_{j,k}] \cdot E[Y_{i,k}]
\end{align}
\end{proof}
\begin{theorem} \label{thm:propsampler}
Given parameters $\tilde{n}$ and $0 < \varepsilon < 1$, $\PropEstimator(\tilde{n}, \varepsilon)$ has sample complexity $O(\frac{\sqrt{\tilde{n}}}{\varepsilon})$ and returns  an estimate $\hat{W}$ such that $E[1/\hat{W}] = 1/W$. If, moreover, $\tilde{n} \geq n$, then $P(|\hat{W} - W| \leq \varepsilon W) \geq 2/3$.
\end{theorem}

\begin{proof}
The sample complexity is clearly as claimed. We now prove that $1 / \hat{W}$ is an unbiased estimator of $1 / W$:
\begin{align}
& \frac{1}{\hat{W}} = \binom{m}{2}^{-1} \sum_{s \in S} \frac{\binom{c_s}{2}}{w(s)} = \binom{m}{2}^{-1} \sum_{1\leq i<j \leq m} Y_{ij}
\end{align}
and thus
\begin{align}
& E\ld[\frac{1}{\hat{W}}\rd] = \binom{m}{2}^{-1} \sum_{1\leq i<j \leq m} E\ld[Y_{ij}\rd] = \frac{1}{W}.
\end{align}
When $i,j,k,\ell$ are all distinct, $Y_{ij}$ and $Y_{k\ell}$ are independent. Moreover, by \Cref{lem:onepair}, $Y_{ij}$ and $Y_{ik}$ are uncorrelated for $j \neq k$. Using the bound on $Var(Y_{ij})$ from \Cref{lem:onepair}, we then have that
\begin{align}
Var\ld[\frac{1}{\hat{W}}\rd] =& \binom{m}{2}^{-2} \sum_{1\leq i<j \leq m} Var(Y_{ij})  \\
\leq& \binom{m}{2}^{-1} \frac{n}{W^2} \\
\leq& \frac{1}{12} \cdot \ld( \frac{\varepsilon}{W}\rd)^2
\end{align}
By Chebyshev inequality, it holds that
\[
P\ld(\bigg|\frac{1}{\hat{W}} - \frac{1}{W}\bigg| > \frac{\varepsilon}{2W}\rd) \leq \frac{Var\ld[1/\hat{W}\rd]}{\ld(\varepsilon/2W\rd)^2}\leq \frac{1}{3}
\]
Finally, for $\varepsilon \leq 1$ we have 
\[
(1-\varepsilon)W \leq (1+\varepsilon/2)^{-1}W \leq \hat{W} \leq (1-\varepsilon/2)^{-1}W \leq (1+\varepsilon)W
\]
This means that $|1/\hat{W} - 1/W| \leq \varepsilon/(2W)$ implies $|\hat{W} - W| \leq \varepsilon W$. Thus $P\ld(|\hat{W} - W| \leq \varepsilon W\rd) \geq 2/3$.

\end{proof}

\subsection{Algorithms for \texorpdfstring{$|U|$}{u} unknown.}
In this section, we present the algorithm $\NAPropEstimator$, which samples elements from $U$ proportionally to their weights, and computes an $\onepmepsilon$-approximation of $W$ with probability $2/3$, without any knowledge of $n=|U|$. 

$\NAPropEstimator$ takes $O(\sqrt{n}/\varepsilon + \log n/\epsilon^2)$ samples in expectation and works as follows. 
We partition $U$ into buckets such that items in one bucket have roughly the same weight. We pick one bucket such that the items in this bucket are likely to have a sufficiently large total weight. We then estimate the sum restricted to this bucket. If we are able to do that, we can estimate the total weight by looking at what fraction of the proportional samples end up in this bucket. We estimate the sum restricted to the bucket as follows. Since the weights are roughly the same for all items in the bucket, we may use rejection sampling to efficiently simulate uniform samples from the bucket. That allows us to estimate the number of items in it, up to a constant factor. We use the algorithm $\PropEstimator$ with this estimated bucket size as the advice $\tilde{n}$.

\subsubsection*{Estimating $|U|$ through uniform sampling.}
As a preliminary step, we assume that we are able to sample elements from $U$ uniformly, rather than according to their weights. Under this assumption, we introduce the algorithm $\unisizeestimator$ that, using $O(\sqrt{n})$ expected samples, estimates $n=|U|$ up to a constant factor with probability $2/3$.
The intuition behind $\unisizeestimator$ is fairly simple: if we sample uniformly with replacement from a universe of size $n$ and we see the first repetition after $t$ samples, then it is likely that $t \approx \sqrt{n}$. 

\begin{algorithm2e}[ht]
\DontPrintSemicolon
$S_0 \leftarrow \emptyset$ \\
\For{$i \in \mathbb{N}$}{
    Sample $a_i \in U$ uniformly \\
    \If{$a_i \in S_i$}{
        $\hat{s} \leftarrow |S_i|$ \\
        $\hat{N} \leftarrow 4 \hat{s}^2$ \\
        \Return{$\hat{N}$}
    }
    $S_{i+1} \leftarrow S_i \cup \{a_i\}$
    }

\caption{$\unisizeestimator()$} \label{alg:set_size_estimation}
\end{algorithm2e}

\begin{theorem} \label{thm:universesizeestimator_analysis}
$\unisizeestimator$ has expected sample complexity of $O(\sqrt{n})$. It returns an estimate $\hat{N}$ such that $P(n \leq \hat{N}) \geq 2/3$ and $E[\hat{N}] = O(n)$.
\end{theorem}
\begin{proof}
We prove that when the algorithm aborts, it holds $P(\sqrt{n}/2 \leq \hat{s}) \geq 2/3$. The bound on $P(n \leq \hat{N})$ follows by the definition of $\hat{N}$.
Define the event $\Ee_i = \{a_i \in S_i\}$, where we define that $a_i \not\in S_i$ whenever the algorithm terminates before step $i$. It then holds $P(\Ee_i) \leq P(\Ee_i \,|\, \bigcap_{j < i} \Bar{\Ee}_j) = i/n$. We have
\begin{align}
P\ld(\hat{s} < \frac{\sqrt{n}}{2}\rd) &= P\ld(\bigcup_{i=0}^{\sqrt{n}/2 - 1} \Ee_i\rd) \\
&\leq \sum_{i=0}^{\sqrt{n}/2 - 1} P(\Ee_i) \\
&\leq \sum_{i=0}^{\sqrt{n}/2 - 1} \frac{i}{n} \\
&= \frac{1}{n} \cdot \binom{\sqrt{n}/2}{2} < \frac{1}{6}.
\end{align}
%
After $\sqrt{n}$ samples, each additional sample is a repetition with probability at least $1/\sqrt{n}$. The number of iterations before the algorithm returns is thus stochastically dominated by $\sqrt{n} + Geom(1/\sqrt{n})$. We may thus bound the expectation as
\[
E\ld[\hat{N}\rd] = E\ld[4 \hat{s}^2\rd] \leq O\ld(n + E\ld[Geom(1/\sqrt{n})^2\rd]\rd) \leq O(n)
\]
where the last inequality is a standard result on the second moment of the geometric random variable.
\end{proof}

\subsubsection*{Simulating uniform sampling.}
We define buckets $B_i = \{ a \in U \,|\, w(a) \in [2^i, 2^{i+1})\}$ for each $i \in \mathbb{Z}$, and we show how to sample elements uniformly from $B_i$, while allowed to sample elements proportionally to their weight $w$.
First, we show how to sample elements from a bucket $B_b$ given a $b \in \mathbb{Z}$  proportionally in $\propbucketsampler$. We then use rejection sampling to obtain a uniform sample through $\unifbucketsampler$.

\begin{algorithm2e}[ht]
\DontPrintSemicolon
Sample $(a, w(a))$ proportionally \\
\While{$w(a) \not\in [2^b, 2^{b+1})$}{
    Sample $(a, w(a))$ proportionally \\
}
\Return{$(a, w(a))$}
\caption{$\propbucketsampler(b)$} \label{alg:prop_bucket_sampler}
\end{algorithm2e}

\begin{algorithm2e}[ht]
\DontPrintSemicolon
$(a, w(a)) \gets \propbucketsampler(b)$ \\
\While{$\textsc{Uniform}([0, 1]) >  \frac{2^b}{w(a)}$}{
    $(a, w(a)) \gets \propbucketsampler(b)$ \\
}
\Return{$(a, w(a))$}
\caption{$\unifbucketsampler(b)$} \label{alg:unif_bucket_sampler}
\end{algorithm2e}

\begin{lemma}\label{lem:bucketsampler}
Let $p_b = P(w(a) \in [2^b, 2^{b+1}))$ for $a\in U$ sampled proportionally. Then, the expected sample complexity of both $\propbucketsampler$ and $\unifbucketsampler$ is $O(1/p_b)$.
$\propbucketsampler$ returns an item from the $b$-th bucket with distribution proportional to the weights. $\unifbucketsampler$ returns an item from the $b$-th bucket distributed uniformly.
\end{lemma}
\begin{proof}
$\propbucketsampler$ performs samples until it samples an item $a$ from bucket $B_b$; it returns $a$. This is equivalent to sampling proportionally conditioned on $a \in B_b$. This proves that the output has the claimed distribution. In each step, we finish with probability $p_b$, independent of other steps. The expected number of steps is, therefore, $1/p_b$. This proves the sample complexity.

Similarly, we terminate $\unifbucketsampler$ after sampling $a \in B_b$ and $\textsc{Uniform}([0,1]) \leq \frac{2^b}{w(a)}$. Sampling until $a \in B_b$ is equivalent to sampling item $a$ from $B_b$ with probability $w(a)/A_b$ where $A_b$ is the total weight of items in bucket $B_b$. Therefore, $a$ is sampled in each step with probability
\[
p_b \cdot \frac{w(a)}{A_b} \cdot \frac{2^b}{w(a)} = \frac{p_b 2^b}{A_b}
\]
Since this probability is the same for all items $a \in B_b$, the resulting distribution is uniform. The rejection probability is upper-bounded by $1/2$. Therefore, $\unifbucketsampler$ also  has expected sample complexity of $O(1/p_b)$
\end{proof}

\subsubsection* {Putting it together: Estimating $W$ without advice.}
Finally, we are ready to show the algorithm $\NAPropEstimator$, that estimates $W$ without relying on any advice $\tilde{n} \geq n$.
To analyze it, we first need a lemma.

\begin{algorithm2e}[ht]
\DontPrintSemicolon
Sample $(a_1, w(a_1)), (a_2, w(a_2))$ proportionally \label{lst:line:asampled} \\
$b_i \leftarrow \lfloor \log w(a_i) \rfloor$ for $i = 1,2$ \\
$b \leftarrow \max(b_1, b_2)$ \label{lst:line:bsampled}\\
$\tilde{n}_b \leftarrow \unisizeestimator()$ using $\unifbucketsampler(b)$ as sampling subroutine, with success probability $9/10$ \label{lst:line:ntilde}\\
$\hat{W}_b \leftarrow \PropEstimator(\frac{\varepsilon}{3}, \tilde{n}_b)$ using $\propbucketsampler(b)$ as sampling subroutine, with success probability $9/10$ \label{lst:line:what}\\
$\hat{P}_b \leftarrow \bernoulliestimation\ld(\frac{\varepsilon}{3}\rd)$ to estimate $P(a \in B_b)$ with success probability $9/10$ \label{lst:line:phat} \\
$\hat{W} \leftarrow \hat{W}_b / \hat{P}_b$ \\
\Return{$\hat{W}$}
\caption{$\NAPropEstimator(\varepsilon)$} \label{alg:noadvice_proportional}
\end{algorithm2e}

To analyze $\NAPropEstimator$, we first need a lemma:
\begin{lemma}\label{lem:seriesofprobs}
Consider $b_1, b_2$ and $b$ as defined in $\NAPropEstimator$, and let $a\in U$ be a random element sampled proportionally. Then,
\[
E\ld[\frac{1}{P_{prop}\ld(a \in B_b \,|\, b\rd)}\rd]
= O(\log n).
\]
Moreover, if we define $n_b$ as the number of items in $B_b$ we have
\[
E\ld[\frac{\sqrt{n_b}}{P_{prop}\ld(a \in B_b \,|\, b\rd)}\rd] = O(\sqrt n).
\].
\end{lemma}

\begin{proof}
Throughout this proof, we assume $a$ to be sampled proportionally. We first prove the first statement. Define $k = \max \{ j\,|\, B_j \neq \emptyset\}$, and set $\Bb = \{B_j | k - 2 \log n < j \leq k \text{ and } B_j \neq \emptyset\}$.
Notice that $|\Bb| \leq 2 \log n$ and, for each $B_j \not\in \Bb$, $x \in B_j$ we have
\[
P\ld(a = x\rd) \leq \frac{2^{k - 2 \log n + 1}}{W} = 2^{-2 \log n + 1} \cdot \frac{2^k}{W} \leq \frac{2}{n^2}
\]
since there exists $y \in B_k$ and therefore $W \geq w(y) \geq 2^k$. Hence, sampling $a \in U$ proportionally we have
\begin{align}
P\ld(B_a \not\in \Bb\rd) &= P\ld(\exists j,\text{ s.t. }a \in B_j \land B_j \not\in \Bb\rd) \\
&\leq \sum_{x \in \bigcup_{B_j \not\in \Bb}B_j} P(a = x) \\
&\leq n \cdot \frac{2}{n^2} = \frac{2}{n}.
\end{align}
Now, we notice that for each $B_j$, it holds
\begin{align} \label{1}
P\ld(b \in B_j\rd) &\leq 2 \cdot P\ld(b_1 \in B_j\rd) \cdot P\ld(\exists i\leq j \,:\, b_2 \in B_i\rd) \\
&= 2 \cdot P\ld(a \in B_j\rd) \cdot P\ld(\exists i\leq j \,:\, a \in B_i\rd) \\
&\leq 2 \cdot P\ld(a \in B_j\rd)
\end{align} 
where the factor two is given by the union bound, and we used $a \sim b_1 \sim b_2$. Then, we can write
\begin{align}
E\ld[\frac{1}{P\ld(a \in B_b \,|\, b\rd)}\rd] &=\sum_{B_j \neq \emptyset} \frac{P\ld(b\in B_j\rd)}{P\ld(a\in B_j\rd)} \\ &= \sum_{B_j \in \Bb} \frac{P\ld(b\in B_j\rd)}{P\ld(a\in B_j\rd)} + \sum_{\substack{B_j \not\in \Bb \\ B_j \neq \emptyset}} \frac{P\ld(b\in B_j\rd)}{P\ld(a\in B_j\rd)} \\
&\leq 2\cdot |\Bb| + \sum_{\substack{B_j \not\in \Bb \\ B_j \neq \emptyset}} 2P\ld(\exists i\leq j \,:\, a \in B_i \rd) \\
&\leq 2\cdot |\Bb| + \sum_{\substack{B_j \not\in \Bb \\ B_j \neq \emptyset}} 2P\ld(B_a \not\in \Bb \rd) \\
&\leq 4 \log n + 2n \cdot P\ld(B_a \not\in \Bb \rd) \\
&\leq 4 \log n + 4 =O(\log n).
\end{align}
The fist inequality is obtained using \cref{1}. The second inequality descends from the fact that $B_j \not\in \Bb$ and $i \leq j$ imply $B_i \not\in \Bb$. The last two inequalities are obtained using $n$ as an upper bound on the number of nonempty buckets $B_j$ and recalling that $P(B_a \not\in \Bb) \leq 2/n$.

Now we can prove the second statement. Denote by $n_b$ the number of elements in $B_b$ and define $\ell = \argmax_{i \in \mathbb{Z}} n_i 2^{i/2}$. If we define $S_i = \sum_{j \leq i} \sum_{a\in B_j} w(a)$, then we can rewrite the already proven inequality $P\ld(b \in B_j\rd) \leq 2 \cdot P\ld(a \in B_j\rd) \cdot P\ld(\exists i\leq j \,:\, a \in B_i\rd)$ as
\begin{equation} \label{2}
P\ld(b \in B_j\rd) \leq 2 \cdot P\ld(a \in B_j\rd) \cdot \frac{S_j}{W}.
\end{equation}
We now prove that there exists a constant $C>0$ such that, for all $i \in \mathbb{Z}$ it holds $S_{\ell - i} \leq C \cdot S_{\ell} \cdot 2^{-i/2}$.
Notice that, by definition of $\ell$, we have $n_j \cdot 2^{j/2} \leq n_{\ell} \cdot 2^{\ell/2}$ for all $j\in \mathbb{Z}$. We can now bound
\begin{align}
    S_{\ell-i} &\leq \sum_{j\leq \ell - i} n_j \cdot 2^{j+1} \\
    &\leq n_{\ell} \cdot \sum_{j\leq \ell - i} 2^{\ell/2 - j/2} \cdot 2^{j+1} \\
    &= 2 n_{\ell} 2^{\ell} \cdot \sum_{j\leq \ell - i} 2^{j/2 - \ell/2}\\
    &\leq 2 S_{\ell} \cdot \sum_{j\leq \ell - i} 2^{j/2 - \ell/2} \leq C \cdot S_{\ell} 2^{-i/2}.
\end{align}
Therefore, we have $\sum_{j < \ell} S_j = O(S_{\ell})$. Notice that by the definition of $\ell$ we have $n_{\ell + i} \leq n_{\ell} \cdot 2^{-i/2}$ for each $i \geq 0$. 
Now we are ready to prove our final result.
\begin{align}
    E\ld[ \frac{\sqrt{n_b}}{P\ld(a \in B_b \,|\, b\rd)}\rd] &= \sum_{j \in \mathbb{Z}} P\ld(b\in B_j\rd) \cdot \frac{\sqrt{n_j}}{P\ld(a\in B_j\rd)} \\
    &\leq \sum_{j \in \mathbb{Z}} 2 \frac{S_j}{W} \cdot \sqrt{n_j} \\
    &\leq \sum_{j < \ell} 2 \frac{S_j}{W} \cdot \sqrt{n} + \sum_{j \geq \ell} 2 \sqrt{n_j} \\
    &\leq \frac{O(S_{\ell})}{W} \cdot \sqrt{n} + \sum_{i \geq 0} 2  
    \sqrt{n_\ell} \cdot 2^{-i/4} = O(\sqrt{n}).
\end{align}
The first inequality uses \cref{2}, the second inequality is obtained splitting the series in two parts and using $S_j \leq W$. The last inequality is obtained plugging in $\sum_{j < \ell} S_j = O(S_{\ell})$ and $n_{\ell + i} \leq n_{\ell} \cdot 2^{-i/2}$ for $i \geq 0$.
\end{proof}

Now we are ready to analyze $\NAPropEstimator$.
\begin{theorem}\label{lem:no_advice_without_expectation_bound}
Let $\hat{W}$ be the estimate returned by $\NAPropEstimator$. Then $\hat{W}$ is an $\onepmepsilon$-approximation of $W$ with probability $2/3$. Moreover, its expected sample complexity is  
\[
O \ld(\frac{\sqrt{n}}{\varepsilon} + \frac{\log (n)}{\varepsilon^2}\rd).
\]
\end{theorem}
\begin{proof}
We start by proving correctness. Define $W_b = \sum_{x \in B_b} w(x)$ and $P_b= P(a \in B_b\,|\, b) = W_b / W$. Notice that $W_b$ and $P_b$ are random variables, since they depend on $b$. Now we prove that $\hat{W}$ is a $\onepmepsilon$-approximation of $W$ with probability $2/3$. Define the event $\Ee_1 = \{n \leq \tilde{n}_b\}$, we have $P(\Ee_1) \geq 9/10$ (where we use \Cref{thm:universesizeestimator_analysis} together with probability amplification to amplify the success probability of $2/3$ to $9/10$). Define the event
\[
\Ee_2 = \ld\{(1-\varepsilon/3) W_b \leq \hat{W}_b \leq (1+\varepsilon/3) W_b \rd\}
\]
then, we have $P(\Ee_2\,|\,\Ee_1) \geq 9/10$ (where we use Theorem~\ref{thm:propsampler} and probability amplification).
 Define the event 
\[
\Ee_3= \ld\{(1-\varepsilon/3) P_b \leq \hat{P}_b \leq (1+\varepsilon/3) P_b \rd\}
\]
then it holds $P(\Ee_3\,|\,b) \geq 9/10$ (where we use \Cref{claim:bernoulli_estimation} and probability amplification). On the event $\Ee_2\cap \Ee_3$, it holds
\[
(1-\varepsilon) \cdot W \leq 
\frac{1-\varepsilon/3}{1+\varepsilon/3} \cdot \frac{W_b}{P_b} \leq
\frac{\hat{W}_b}{\hat{P}_b} \leq
\frac{1+\varepsilon/3}{1-\varepsilon/3} \cdot \frac{W_b}{P_b} \leq 
(1+\varepsilon) \cdot W.
\]
Then we can apply union bound and prove
\begin{align}
P\ld(\hat{W} < (1 - \varepsilon) W \text{ or } \hat{W} > (1+ \varepsilon)W\rd) &\leq \\
P\ld(\Bar{\Ee}_2 \cup \Bar{\Ee}_3\rd) &\leq \\
P\ld(\Bar{\Ee}_1\rd) + P\ld(\Bar{\Ee}_2 \,|\, \Ee_1\rd) + P(\Bar{\Ee_3}) &\leq \\
\frac{1}{10} + \frac{1}{10} + \frac{1}{10} &\leq \frac{1}{3}.
\end{align}
It remains to prove that the expected number of samples that $\NAPropEstimator$ uses is as claimed.
Denote by $\sigma_1$ the total number of samples taken on \cref{lst:line:ntilde}, by $\sigma_2$ the total number of samples taken on line~\ref{lst:line:what} and by $\sigma_3$ the total number of samples taken on line~\ref{lst:line:phat}. We denote the number of samples employed during the $i$-th call to $\unifbucketsampler(b)$ on line~\ref{lst:line:ntilde} with $\eta_1^{(i)}$; similarly, we denote the number of samples taken during the $i$-th call to $\propbucketsampler(b)$ on line~\ref{lst:line:what} with $\eta_2^{(i)}$.
We can then write
\[
\sigma_1 = \sum_{i=1}^{\tau_1} \eta_1^{(i)}\, \text{ and }\, \sigma_2 = \sum_{i=1}^{\tau_2} \eta_2^{(i)}
\]
where $\tau_1$ and $\tau_2$ are the number of calls to $\unifbucketsampler(b)$ performed on line~\ref{lst:line:ntilde} and line~\ref{lst:line:what}, respectively. First we notice that, thanks to Lemma~\ref{lem:bucketsampler}, there exists a constant $K > 0$ such that $E[\eta^{(i)}_1\,|\, b], E[\eta^{(i)}_2\,|\, b] \leq \frac{K}{P(a \in B_b\,|\, b)}$. Thanks to Theorem~\ref{thm:universesizeestimator_analysis}, we have $E[\tau_1 \,|\, b] = O(\sqrt{n_b})$ and $\tau_2 = O(\sqrt{\tilde{n_b}} / \varepsilon)$.
Now we are ready to bound $E[\sigma_1]$ and $E[\sigma_2]$. We have
\begin{align}
E\ld[\sigma_1\rd] &= E\ld[E\ld[\sum_{i=1}^{\tau_1} \eta_1^{(i)} \,\big|\, b \rd] \rd] \\
&= E\ld[E\ld[\tau_1\,\big|\, b\rd] \cdot E\ld[\eta^{(1)}_1\,\big|\,b \rd]\rd] \\
&\leq E\ld[\frac{K \cdot O(\sqrt{n_b})}{P(a \in B_b\,|\, b)}\rd] = O\ld(\sqrt{n}\rd)
\end{align}
where the first equality is by the Wald's identity and the last equality is obtained applying Lemma~\ref{lem:seriesofprobs}. Similarly,
\begin{align}
E\ld[\sigma_2\rd] &= E\ld[E\ld[\sum_{i=1}^{\tau_2} \eta_2^{(i)} \,\big|\, b \rd] \rd] \\
&= E\ld[E\ld[\tau_2\,\big|\, b\rd] \cdot E\ld[\eta^{(1)}_2\,\big|\,b \rd]\rd] \\
&\leq E\ld[\frac{K \cdot \ld(\sqrt{80 \tilde{n}_b} / \varepsilon + 1 \rd)}{P(a \in B_b\,|\, b)}\rd] \\
&= O\ld(\frac{\sqrt{n}}{\varepsilon} + \log n\rd) = O\ld(\frac{\sqrt{n}}{\varepsilon} \rd)
\end{align}
where we used that $E[\sqrt{\tilde{n}_b}] \leq \sqrt{E[\tilde{n}_b]} \leq \sqrt{n}$, which holds thanks to \Cref{thm:universesizeestimator_analysis} and Jensen inequality.
In order to bound $E[\sigma_3\,|\, b]$, recall that, thanks to \Cref{claim:bernoulli_estimation}, there exists a $C>0$ such that, conditioning on the value of $b$, $\bernoulliestimation(\frac{\varepsilon}{3})$ takes in expectation at most $\frac{C}{P\ld(a \in B_b \,|\, b\rd) \varepsilon^2}$ samples in order to estimate $P(a \in B_b)$. Therefore we have
\begin{align}
E[\sigma_3] = E\ld[E[\sigma_3\,|\, b]\rd] \leq E\ld[ \frac{C}{P\ld(a \in B_b \,|\, b\rd) \varepsilon^2}\rd]
= O\ld(\frac{\log n}{\varepsilon^2}\rd).
\end{align}
where the last equality holds by \Cref{lem:seriesofprobs}. This concludes the proof, since the total number of samples taken by $\NAPropEstimator$ is $\sigma_1 + \sigma_2 + \sigma_3$. By the bounds we have proven above, the expectation of $\sigma_1 + \sigma_2 + \sigma_3$ is as claimed.
\end{proof}

\section{Sum Estimation by Hybrid Sampling.} \label{sec:hybrid_sampling}
In this section, we assume that we can sample elements both proportionally and uniformly. Again, we solve the task of providing an estimate $\hat{W}$ of $W$ such that $\hat{W}$ is a $\onepmepsilon$-approximation of $W$ with probability $2/3$.

We notice that if we take $\Theta(n \log n)$ uniform samples then, with probability $2/3$, we see every element of $U$. This simple analysis is well-known under the name of Coupon Collector problem. If we know $n$ (or any constant-factor approximation of it) we may simply take $\Theta(n \log n)$ samples, assume we have seen every element at least once, and compute $W$ exactly. However, if we do not know $n$, a more complex scheme is required to achieve a complexity of $O(n\log n)$, which we describe in \Cref{sec:hybrid_unknownn}. Therefore, it is sufficient to show an algorithm with complexity $T(n, \varepsilon)$ to obtain a complexity of the form $O(\min(T(n, \varepsilon), n \log n))$, as we can just run the coupon-collector algorithm in parallel and take the result provided by the first of the two algorithms to finish its execution. In what follows we only show how to achieve a complexity of $O(\sqrt[3]{n}/\varepsilon^{4/3})$ when $|U|$ is known, and of $O(\sqrt{n}/\varepsilon)$ when $|U|$ is unknown. As a consequence, the complexities that we achieve in this settings are $O(\min(\sqrt[3]{n}/\varepsilon^{4/3}, n \log n))$  and of $O(\min(\sqrt{n}/\varepsilon, n \log n))$ respectively. 

\subsection{Algorithms for \texorpdfstring{$|U|$}{u} known.} \label{sec:hybrid_knownn}
In this section, we show an algorithm that, given $n=|U|$, returns a $\onepmepsilon$-approximation of $W$ with probability $2/3$ using $O(\sqrt[3]{n}/\varepsilon^{4/3})$ samples. First, we introduce a subroutine that uses harmonic mean to estimate the average weight $W/n$; then we combine it with $\PropEstimator$ to obtain the main algorithm of this section. 

\subsubsection*{Harmonic-mean-based estimator.} 
Here we show the algorithm $\harmonicestimator(\varepsilon, \tilde{\theta}, \phi)$ that returns an $(1 \pm \varepsilon)$-approximation $\hat{\theta}$ of $W / n$ with probability $2/3$.
$\harmonicestimator(\varepsilon, \tilde{\theta}, \phi)$ takes as advice an upper bound on the average weight $\tilde{\theta} \geq W / n$, and a parameters
$\phi$ such that we expect $\punif(w(a) \geq \phi)$ not to be too small and $\tilde{\theta} / \phi$ not to be too large. A more formal statement follows.

\begin{algorithm2e}[ht]
$\hat{p}\leftarrow \bernoulliestimation\ld(\punif(w(a) \geq \phi), \frac{\varepsilon}{3}\rd)$ with success probability $9/10$ \label{lst:line:harmonicbernoulli} \\
$k \leftarrow 45 \cdot \frac{\tilde{\theta}}{\phi (1-\varepsilon/3)\hat{p} \varepsilon^2}$ \\ 
Sample $a_1 \dots a_k$ proportionally \label{lst:line:harmonicprop}\\\
\For{$i=1 \dots k$}{
\If{$w(a_i) \geq \phi$}{
$b_i = 1 / w(a_i)$ \\
}
\Else{
$b_i = 0$ \\
}
}
$H = \sum_{i=1}^k b_i / k$ \\
$\hat{\theta} \leftarrow \hat{p} / H$ \\
\Return{$\hat{\theta}$}
\caption{$\harmonicestimator(\varepsilon, \tilde{\theta}, \phi)$} \label{alg:estimate_by_harm_mean}
\end{algorithm2e}
To see the intuition behind this algorithm, consider the case when $\phi \leq w(a)$ for all $a \in U$. It then holds $\hat{p} \approx 1$. We take $k$ samples, and let $1/H$ be the harmonic mean of the weights of the sampled items. We have $E[H]=n/W$, and $\hat{\theta} \approx 1/H$ as $\hat{p} \approx 1$. Unfortunately $H$ might have a high variance due to elements having very small weights. To fix this, we consider a parameter $\phi$ such that $w(a) < \phi$ for some $a \in U$. Instead of $E[H] = n/W$, we then have $E[H] = n'/W$ for $n' = |\{a \in U \,|\, w(a) \geq \phi \}|$. We then multiply $1/H$ by $\hat{p}$ in order to adjust for the fraction of items that were ignored. Note that, while increasing $\phi$, the variance of $1/H$ decreases; however, also $n'$ decreases and this means that computing an estimate $\hat{p} \approx n'/n$ requires more samples. This introduces a trade-off between the algorithm's complexity and the variance of $H$.

\begin{lemma} \label{lem:estimate_by_harm_mean}
Given parameters $\tilde{\theta}$ and $0 < \varepsilon < 1$, $\harmonicestimator(\varepsilon, \tilde{\theta}, \phi)$ has expected sample complexity $O((1 +\frac{\tilde{\theta}}{\phi}) / (p \, \varepsilon^2))$ where $p=\punif(w(a) \geq \phi)$. It returns an estimate $\hat{W}$ such that $P(\hat{\theta} < W/(20n)) \leq 1/20$. If, moreover, $\tilde{\theta} \geq W/n$, then $P(|\hat{W} - W| \leq \varepsilon W) \geq 2/3$.
%
\end{lemma}
\begin{proof}
We start by proving that $\hat{\theta}$ is a $(1+\varepsilon)$-approximation of $W/n$ with probability $2/3$ when $\tilde{\theta} \geq W/n$.
Define the event $\Ee=\{\hat{p} \text{ is a } (1\pm \varepsilon/3)\text{-approximation of } p\}$. By \Cref{claim:bernoulli_estimation} and using probability amplification, we have $P(\Ee) \geq 9/10$. 
For each $i = 1 \dots k$ we have
\[
E[b_i] = \sum_{\substack{a \in U, \\ w(a) \geq \phi}} \frac{1}{w(a)} \cdot \frac{w(a)}{W} = \frac{n_{\geq \phi}}{W} = \frac{p \cdot n}{W}
\]
where $n_{\geq \phi}$ is the number of elements in $a\in U$ with $w(a) \geq \phi$. Notice that this implies $E[H] = p \cdot n/W$. Moreover, for each $i = 1 \dots k$
\begin{align}
Var(b_i) \leq E\ld[b_i^2\rd] &=\\
\smash{\sum_{\substack{a \in U, \\w(a) \geq \phi}}} \frac{1}{w^2(a)} \cdot \frac{w(a)}{W} &\leq \\
\frac{n_{\geq \phi}}{\phi \cdot W} &= \frac{p \cdot n}{\phi \cdot W}.
\end{align}
Conditioning on $\Ee$, we have that $(1-\varepsilon/3)\hat{p} \leq p$. The way we have set $k$ allows us to bound 
\begin{align}
Var(H\,|\, \Ee) = \frac{Var(b_i)}{k} &= \\
\frac{p \cdot n}{\phi \cdot W} \cdot \frac{\varepsilon^2 (1-\varepsilon/3)\hat{p}}{45} \cdot \frac{\phi}{\tilde{\theta}} &\leq  \\
\ld(\frac{p \cdot n}{W}\rd)^2 \cdot \frac{\varepsilon^2}{45} &= E[H]^2 \cdot \frac{\varepsilon^2}{45}
\end{align}
where we used that $\tilde{\theta} \geq W/n$.
It holds $E[H \,|\, \Ee] = E[H]$. We may thus apply Chebyshev's inequality to get
\begin{align}
P\ld(|H - E[H]| > \frac{\varepsilon}{3} E[H] \,\bigg|\, \Ee\rd) \leq \frac{Var(H\,|\,\Ee)}{\ld(\frac{\varepsilon}{3} E[H]\rd)^2} \leq \frac{9}{45} = \frac{1}{5}
\end{align}
Since $\varepsilon < 1$, we have $(1-\varepsilon/3)^{-1} \leq 1+\varepsilon/2$ and $(1+\varepsilon/3)^{-1} \geq 1-\varepsilon/2$. Therefore
\[
P\ld(\bigg|\frac{1}{H} - \frac{1}{E[H]}\bigg| > \frac{\varepsilon}{2} \cdot \frac{1}{E[H]} \,\bigg|\, \Ee\rd) \leq \frac{1}{5}.
\]
Again, since $\varepsilon < 1$, we have $(1 + \varepsilon /3) \cdot (1 + \varepsilon/2) \leq 1 + \varepsilon$ and $(1 - \varepsilon /3) \cdot (1 - \varepsilon/2) \geq 1 - \varepsilon$. Hence, using the union bound 
\begin{align}
P\ld(\Bigg|\frac{\hat{p}}{H}-\frac{p}{E[H]}\Bigg| > 
\varepsilon\cdot \frac{p}{E[H]}\rd) &\leq \\ 
P(\Bar{\Ee}) + P\ld(\bigg|\frac{1}{H} - \frac{1}{E[H]}\bigg| > \frac{\varepsilon}{2} \cdot \frac{1}{E[H]} \,\bigg|\, \Ee\rd) &\leq \frac{1}{3}.
\end{align}
Since $p / E[H] = W/ n$, we have that the estimate $\hat{\theta}= \hat{p}/H$ is a $(1+\varepsilon)$-approximation of $W/n$ with probability $\geq 2/3$.

We now argue the sample complexity. The expected number of samples used on line~\ref{lst:line:harmonicbernoulli} is by \Cref{claim:bernoulli_estimation} equal to $O(1/(p\,\varepsilon^2))$. In the rest of the algorithm, we use $k$ samples. It holds
\[
E\ld[k\rd] = E\ld[\frac{1}{\hat{p}}\rd] \cdot O\ld(\frac{\tilde{\theta}}{\phi\, \varepsilon^2}\rd) = O\ld(\frac{\tilde{\theta}}{p\,\phi\, \varepsilon^2}\rd)
\]
where the second equality holds by \Cref{claim:bernoulli_estimation}. The sample complexity is thus as claimed.

Finally, we prove that, regardless of $\tilde{\theta}$, it holds $P(\hat{\theta} < W/(20n)) \leq 1/20$. It holds $E[H/\hat{p}] = E[H] E[1/\hat{p}] = n/W$ where $E[1/\hat{p}]=1/p$ by \Cref{claim:bernoulli_estimation}. Therefore, by the Markov's inequality 
\[
P\ld(\frac{\hat{p}}{H} \leq \frac{W}{20n}\rd) = P\ld(\frac{H}{\hat{p}} \geq \frac{20n}{w}\rd) \leq \frac{1}{20}.
\]
\end{proof}

\subsubsection*{Combining the two algorithms.}
Here, we combine $\harmonicestimator$ with $\PropEstimator$ to obtain $\HybridEstimator$. It works in the hybrid setting and provides a $\onepmepsilon$-approximation of $W$ with probability $2/3$, using in expectation $O(\sqrt[3]{n}/\varepsilon^{4/3})$ samples. While analysing $\HybridEstimator$, we can restrict ourselves to  $\varepsilon \geq 8 / \sqrt{n}$. Indeed, for very small values of $\varepsilon$ (namely, $\varepsilon \leq 1/(\sqrt{n}\log n)$) we use the coupon-collector algorithm, and for intermediate values of $\varepsilon$ (namely, $ 1/(\sqrt{n}\log n) < \varepsilon < 8 / \sqrt{n}$) we use $\PropEstimator(n, \varepsilon)$. The coupon collector algorithm gives a sample complexity of $O(n\log n)$, that is better than $O(\sqrt[3]{n}/\varepsilon^{4/3})$ for $\varepsilon < 1/(\sqrt{n} \log n)$. $\PropEstimator$ gives a sample complexity of $O(\sqrt{n}/\varepsilon)$, that is better than $O(\sqrt[3]{n}/\varepsilon^{4/3})$ for $\varepsilon < 8 / \sqrt{n}$. 

\begin{algorithm2e}[ht]
\DontPrintSemicolon 
\nonl Abort this algorithm if it uses more than $C\, \frac{n^{1/3}}{\varepsilon^{4/3}}$ samples, where $C$ is a large enough constant. \\
\medskip
Find $\theta$ such that $P\ld(\frac{n^{2/3}}{\varepsilon^{2/3}} \leq \bigg| \ld\{a \in U \,\big|\, w(a) \geq \theta \rd\}\bigg| \leq 2 \cdot \frac{n^{2/3}}{\varepsilon^{2/3}} \rd) \geq \frac{19}{20}$ \label{lst:line:findtheta}\\
$\hat{p}\leftarrow \bernoulliestimation\ld(\pprop(w(a) \geq \theta), \frac{\varepsilon}{3}\rd)$ with success probability $19/20$ \label{lst:line:bernoullicombo}\\
\eIf{$\hat{p} \geq 1/2$}{
$\tilde{n}_{\geq \theta} \leftarrow 2 \cdot \frac{n^{2/3}}{\varepsilon^{2/3}}$ \\
    $\hat{W}_{\geq \theta} \leftarrow \PropEstimator\ld(\tilde{n}_{\geq \theta}, \frac{\varepsilon}{3}\rd)$ with success probability $19/20$, run on proportional samples conditioned on $w(a) \geq \theta$, obtained by rejecting elements with $w(a) < \theta$ \label{lst:line:filteredcombo}\\
    $\hat{W} \leftarrow \hat{W}_{\geq \theta} / \hat{p}$ \\
}{
    $\hat{\rho}\leftarrow \harmonicestimator(\varepsilon, 3\theta, \theta)$ with success probability $19/20$ \label{lst:line:harmonicinhybrid}\\
    $\hat{W} \leftarrow  n \cdot \hat{\rho}$
}
\Return{$\hat{W}$}

\caption{$\HybridEstimator(n, \varepsilon)$} \label{alg:combo_estimation}
\end{algorithm2e}

To implement line~\ref{lst:line:findtheta} it is sufficient to sample uniformly $120 \,n^{1/3}\varepsilon^{2/3}$ elements of $U$ and define $\theta$ as the element with the $180$-th largest weight. (Note that for $\varepsilon \geq 8 / \sqrt{n}$ we have $120 \,n^{1/3}\varepsilon^{2/3} \geq 480$, so this is well-defined.) Let $k$ be such that there are $k$ elements $a \in U$ with $w(a) \geq \theta$. A standard analysis using Chebyshev inequality shows that $k$ is concentrated around $\frac{3}{2} \frac{n^{2/3}} {\varepsilon^{2/3}}$.


\begin{theorem}
Given $\varepsilon$ such that $8 / \sqrt{n} \leq \varepsilon< 1$, $\HybridEstimator(n, \varepsilon)$ uses $O(\sqrt[3]{n}/\varepsilon^{4/3})$ samples. With probability at least $2/3$, the returned estimate $\hat{W}$is a $\onepmepsilon$-approximation of $W$. 
\end{theorem}
\begin{proof}
We first analyze a variant of the algorithm that does not abort. (That is, an algorithm with the same pseudocode except for the abortion condition removed.) We now show that this algorithm returns a $(1\pm\varepsilon)$-approximation with probability at least $5/6$.

The scheme summarized above to find $\theta$ at \cref{lst:line:findtheta} succeeds with probability at least $19/20$. We call this event $\mathcal{E}_1$. Define the event $\Ee_2 = \{\hat{p} \text{ is a } (1\pm \varepsilon/3)\text{-approximation  of } p\}$. It holds $P(\Ee_2) \geq 19/20$.

We now consider the case $\hat{p} \geq 1/2$. On line~\ref{lst:line:filteredcombo} we employ the algorithm $\PropEstimator$. Whenever it performs a sample, we simulate a proportional sample from the set $U_{\theta} = \{a \in U \,|\, w(a) \geq \theta\}$ by sampling until we sample item $a$ such that $w(a) \geq \theta$. 
It is easy to see that the distribution obtained with this sampling scheme is exactly the proportional distribution on the set $U_{\theta}$. Conditioning on $\Ee_1$, $\tilde{n}_{\geq \theta}$ is a valid advice and (by \Cref{thm:propsampler}) $\PropEstimator$ returns a $(1\pm\varepsilon/3)$-approximation of $W_{\theta}=\sum_{w(a) \geq \theta} w(a)$ with probability at least $2/3$. We amplify this probability to $19/20$.
Hence, we have
\begin{align}
P\ld(\ld\{\hat{W}_{\theta} \text{ is a } \onepmepsilon\text{-approximation of } W_{\geq \theta}\rd\} \cap \Ee_2 \rd) &\geq  \\
1 - \ld(\frac{1}{20} + P(\Bar{\Ee_1}) + P(\Bar{\Ee_2})\rd) &\geq \frac{5}{6}.
\end{align}
On this event, since $\varepsilon<1$, we have
\[
(1-\varepsilon)W \leq \frac{1-\varepsilon/3}{1+\varepsilon/3} W \leq \hat{W} \leq \frac{1+\varepsilon/3}{1-\varepsilon/3} W \leq (1+\varepsilon)W.
\]
Now we analyse the case $\hat{p} < 1/2$. Whenever $3\theta \geq W/n$, $\harmonicestimator(\varepsilon, 3\theta, \theta)$ returns a $(1\pm \varepsilon)$-approximation of $W/n$ with probability $2/3$ (by \Cref{lem:estimate_by_harm_mean}).  We amplify that probability to $19/20$. 
Now we argue that, conditioning on $\Ee_2$, we have $3\theta \geq W/n$. Define $p = P_{prop}(w(a) \geq \theta)$, then (on $\Ee_2$) we have $p \leq (1+\epsilon/3)\hat{p} \leq (1+1/3)1/2 \leq 2/3$. Hence, $n \theta \geq (1-p) W \geq W/3$ and thus $3\theta \geq W/n$, where the first inequality is obtained using $\sum_{w(a) < \theta} w(a) = (1-p) W$.
Applying union bound gives that $\harmonicestimator(\varepsilon, 3\theta, \theta)$ succeeds with probability at least $1 - (1/20 + P(\Bar{\Ee_2}))) \geq 5/6$.

Therefore, regardless of the value of $\hat{p}$, we have shown that $\hat{W}$ is a  $\onepmepsilon$-approximation of $W$ with probability at least $5/6$. Thus, the modified algorithm without abortion is correct with probability at least $5/6$. We now argue that the probability that \Cref{alg:combo_estimation} is aborted is at most $1/6$ (for $C$ large enough). By the union bound, its success probability is at least $2/3$.

In what follows, we compute how many samples are taken on each line. 
On line~\ref{lst:line:findtheta}, we use only $120n^{1/3}\varepsilon^{2/3}$ samples. 
Thanks to \Cref{claim:bernoulli_estimation}, $\bernoulliestimation$ on line~\ref{lst:line:bernoullicombo} uses $O(1/(p\varepsilon^2))$ samples in expectation. In what follows, we condition on $\Ee_1 \cap \Ee_2$.
It holds $p \geq P_{unif}(w(a) \geq \theta)$ and thus we have $p \geq \frac{1}{\varepsilon^{2/3} n^{1/3}}$. The number of samples used on \cref{lst:line:bernoullicombo} is then in expectation $O(n^{1/3}/\varepsilon^{4/3})$. By the (conditional) Markov's inequality the probability that we use more than $C_1 n^{1/3}/\varepsilon^{4/3}$ is at most $1/30$, for some $C_1$ large enough.
$\PropEstimator$ uses $O(\sqrt{n^{2/3}/\epsilon^{2/3}}/\epsilon) = O(n^{1/3}/\varepsilon^{4/3})$ samples in the worst case. The rejections cause only constant factor expected slowdown. Again, by the (conditional) Markov's inequality, for $C_2$ large enough, we use more than $C_2 n^{1/3}/\varepsilon^{4/3}$ with probability at most $1/30$.
Since we have $p \geq \frac{1}{\varepsilon^{2/3} n^{1/3}}$, on line~\ref{lst:line:harmonicinhybrid} $\harmonicestimator$ takes $O(n^{1/3}/\varepsilon^{4/3})$ samples in expectation (by \Cref{lem:estimate_by_harm_mean}). Thus, there exists a constant $C_3$ such that on line~\ref{lst:line:harmonicinhybrid} we use more than $C_3 n^{1/3}/\varepsilon^{4/3}$ samples with probability at most $1/30$.

Set $C = 120 + C_1 + C_2 + C_3$. It then holds by the union bound that we use more than $C \sqrt[3]{n}/\varepsilon^{4/3}$ samples (and thus abort) with probability at most $P(\Bar{\Ee}_1) + P(\Bar{\Ee}_2) + 1/30 + 1/30 = 1/6$.
\end{proof}

\subsection{Algorithms for \texorpdfstring{$|U|$}{u} unknown.} \label{sec:hybrid_unknownn}
In this section we show an algorithm that, without any knowledge of $n=|U|$, provides a $\onepmepsilon$-approximation of $W$ with probability $2/3$ using $O(\sqrt{n}/\varepsilon)$ samples. This complexity is strictly higher than the one of \Cref{sec:hybrid_knownn} for $\varepsilon = \omega(1/\sqrt{n})$. However, it is near-optimal when we do not know $n$, as we will see in \Cref{sec:lower_bounds_section}.

\begin{algorithm2e}[ht]\label{alg:noadvice_hybrid}
$\tilde{n} \leftarrow \unisizeestimator()$, with success probability $5/6$ (using uniform sampling) \\
$\hat{W} \leftarrow \PropEstimator(\tilde{n}, \varepsilon)$, with success probability $5/6$ (using proportional sampling) \\
\Return{ $\hat{W}$}
\caption{$\NAHybridEstimator(\varepsilon)$}
\end{algorithm2e}

\begin{theorem}
$\NAHybridEstimator(\varepsilon)$ uses in expectation $O(\sqrt{n}/\varepsilon)$ samples and, with probability at least $2/3$, we have $(1-\varepsilon)W \leq \hat{W} \leq (1+\varepsilon)W$.
\end{theorem}
\begin{proof}
By \Cref{thm:universesizeestimator_analysis}, $\unisizeestimator$ takes $O(\sqrt{n})$ samples and returns $\tilde{n}$ such that $n \leq \tilde{n}$ with probability $2/3$. 
We amplify this probability to $5/6$. Conditioning on $n \leq \tilde{n}$, $\PropEstimator$ returns $(1\pm\varepsilon)$-estimate of $W$ with probability at least $2/3$ by \Cref{thm:propsampler}. 
We amplify this probability to $5/6$. By the union bound, the algorithm returns a $(1\pm \varepsilon)$-estimate with probability at least $2/3$.

Moreover, by \Cref{thm:universesizeestimator_analysis} and \Cref{lemma:prob_amplification_expectations}, $E[\tilde{n}] = O(n)$. By the Jensen inequality, $\PropEstimator$ then uses in expectation $O(\sqrt{n}/\varepsilon)$ samples.
\end{proof}

\subsubsection*{Coupon collector algorithm.}\label{sec:coupon_collector}
In this section, we give an algorithm that returns $W$ \emph{exactly} in time $O(n \log n)$, with probability at least $2/3$. In fact, we show that we may learn the whole set $U$, along with all the weights, with probability $2/3$ in this sample complexity. From this, the sum can be easily computed. We can use this to ensure we never use more than $O(n\log n)$ samples in the hybrid setting. Specifically, we may execute in parallel this algorithm together with one of the above algorithms and abort when one of them returns an estimate.
If one wishes to implement this in practice, it is possible to instead do the following. Given the parameter $\varepsilon$, compute a threshold $n_0$ such that we would like to (if we knew $n$) execute $\NACouponCollector$ if $n \leq n_0$ and $\NAHybridEstimator$ if $n > n_0$. We then run $\NACouponCollector$ and we abort if in case we find $n_0$ distinct elements. If this happens, we then run $\NAHybridEstimator$.

\begin{algorithm2e}[ht]
$S \leftarrow \emptyset$\\
$k = 0$\\
\While{$k < 4|S| \log 3|S|$}{
    Sample $a \in U$ uniformly \\
    \If{$k \not \in S$}{
        $k \leftarrow 0$\\
        $S \leftarrow S \cup \{a\}$
    }
    \Else{
    $k \leftarrow k + 1$ \\
    }
}
\Return{$\sum_{a \in S} w(a)$}
\caption{$\NACouponCollector()$} \label{alg:coupon collector}
\end{algorithm2e}

\begin{theorem}
$\NACouponCollector$ has expected sample complexity $O(n \log n)$ and returns, with probability at least $2/3$ the whole set $U$.
\end{theorem}
\begin{proof}
Thanks to the analysis of the standard coupon collector problem, we know that it takes in expectation $O(n \log n)$ samples before $S = U$. After that, we spend no more than $4|S| \log 3 |S| = n \log n$ additional samples. The sample complexity is thus as claimed. It remains to argue correctness.

The returned result is not correct if the algorithm returns $S$ too early. This happens exactly if there exists some $\ell$ such that it takes more than $4\ell \log 3\ell$ samples to get the $(\ell+1)$-th element. At step $\ell$, there are $n - \ell$ elements that are not in $S$, hence the probability that none of the $4\ell \log 3\ell$ elements fall in the set $U \setminus S$ is
\[
(\ell/n)^{4 \ell \log 3\ell} \leq (1 - \frac{1}{2 \ell})^{4\ell \log 3\ell} \leq (1/e)^{2 \log 3\ell} = \frac{1}{9 \ell^2}.
\]
The first inequality holds since for $\ell  \in [1, n-1]$, it holds $\ell / n \leq 1 - 1/2\ell$. Taking the union bound over all $1 \leq \ell \leq n-1$, we have that the failure probability is upper-bounded by $\sum_{\ell=1}^{\infty} \frac{1}{9\ell^2} < 1/3$.
\end{proof}

\section{Lower Bounds.} \label{sec:lower_bounds_section}

In this section, we give lower bounds for the problems we study in this paper.
The proofs of the lower bounds all follow a common thread. 
In what follows, we use the term \emph{risk} to refer to the misclassification probability of a classifier.

The roadmap of all our proofs follows. First, we define two different instances of the problem $(U_1, w_1)$ and $(U_2, w_2)$ such that a $\onepmepsilon$-approximation of $W$ is sufficient to distinguish between them. Then, we define our \emph{hard} instance as an equally likely mixture of the two, namely we take $(U_1, w_1)$ with probability $1/2$ and $(U_2, w_2)$ otherwise. We denote these events by $\Ee_1$ and $\Ee_2$ respectively. 
Second, we show that a Bayes classifier\footnote{Suppose we have a partition of the probability space $\Omega$ into events $\Ee_1, \cdots, \Ee_k$. We want to guess which event happened, based on observation $X$. Bayes classifier outputs as its guess $\Ee_\ell$ that maximizes $P(\Ee_\ell | X)$.} cannot distinguish between the two cases with probability $2/3$ using too few samples; since Bayes classifiers are risk-optimal this implies that no classifier can have risk less than $2/3$ while using the same number of samples. 
To show that a Bayes classifier has a certain risk, we study the posterior distribution. Let $S$ represent the outcome of the samples. Since the prior is uniform, applying Byes theorem gives
\[
\frac{P\ld(\Ee_1 \,|\, S\rd)}{P\ld(\Ee_2 \,|\, S\rd)} = \frac{P\ld( S \,|\,\Ee_1 \rd)}{P\ld(S \,|\,\Ee_2\rd)}.
\]
We call this ratio $\Rr(S)$ and show that $\Rr(S) \approx 1$ with probability close to $1$. When $\Rr(X) \approx 1$, the posterior distribution is very close to uniform, and this entails a risk close to $1/2$. First, we show this formally with some technical lemmas, and then we instantiate our argument for each of the studied problems.

\begin{lemma}\label{lem:bayesclas}
Given two disjoint events $\Ee_1, \Ee_2$ such that $P(\Ee_1)=P(\Ee_2)=1/2$ and a random variable $X$, we define the ratio
\[
\Rr(X) = \frac{P\ld( X \,|\,\Ee_1 \rd)}{P\ld(X \,|\,\Ee_2\rd)}.
\]
Notice that $\Rr(X)$ is a random variable since it depends on the outcome of $X$. If
\[
P\ld(\frac{7}{8} \leq \Rr(X) \leq \frac{8}{7} \rd) \geq \frac{14}{15}
\]
then any classifier taking $X$ and classifying $\Ee_1,\Ee_2$ has risk $\geq 2/5$.
\end{lemma}
\begin{proof}
First, we notice that since Bayes classifiers are risk-optimal, then it is sufficient to prove our statement for a Bayes classifier. Define $p_i = P(\Ee_i \,|\, X)$ for $i=1, 2$, then Bayes classifier simply returns $\argmax_{i=1, 2} p_i$. By Bayes theorem and because $P(\Ee_1)=P(\Ee_2)$ we have $\Rr(X) = p_1 / p_2$. If $7/8 \leq \Rr(X) \leq 8/7$, then
\[
\frac{1}{p_2} = \frac{p_1 + p_2}{p_2} = 1 + \Rr(X) \in \ld[\frac{15}{8}, \frac{15}{7}\rd] 
\]
and the same holds for $p_1$, therefore $7/15 \leq p_1, p_2 \leq 8/15$. Hence, conditioning on $7/8 \leq \Rr(X) \leq 8/7$, the probability of correct classification is at most $8/15$. Finally, we have
\begin{align}
P\ld(\text{Classifier returns correct answer}\rd) &\leq \\
P\ld(\Rr(X) < \frac{7}{8} \lor \Rr(X) > \frac{8}{7} \rd) + P\ld(\text{Classifier returns correct answer} \,\bigg|\, \frac{7}{8} \leq \Rr(X) \leq \frac{8}{7}\rd) &\leq \\
\frac{1}{15} + \frac{8}{15} &= \frac{3}{5}. 
\end{align}
\end{proof}

Before proving the main theorems, we need several lemmas.

\begin{lemma}\label{lem:occupationCheby}
Consider an instance of our problem $(U, w)$ so that $n=|U|$ and $w(a) = 1$ for each $a \in U$. We perform $m$ independent samples and denote with $\ell$ the number of distinct elements obtained. Then, $Var(\ell)\leq m^2/n$. Moreover, for each $\delta > 0$ 
\[
P\ld(\big|\ell - E[\ell]\big| \geq \frac{1}{\sqrt{\delta}} \cdot \frac{m}{\sqrt{n}}\rd) \leq \delta.
\]
\end{lemma}
Note that, since all weights are the same, the proportional sampling is equivalent to sampling uniformly from $U$.
\begin{proof}
We use the Efron-Stein inequality to prove a bound on the variance of $\ell$. 
Let $X_i$ be the $i$-th sample. Now $\ell$ is a function of $X_1, \cdots, X_m$ and we write it as $\ell = f(X_1, \cdots, X_m)$. Let $X_1', \cdots, X_m'$ be an independent copy of $X_1, \cdots, X_m$. Let $\ell_i' = f(X_1, \cdots, X_{i-1}, X_i', X_{i+1}, \cdots, X_m)$. It clearly holds that $|\ell - \ell_i'| \in \{0,1\}$, moreover $\ell - \ell_i' = 1$ if and only if $X_i$ does not collide with with any $X_j$ for $j \neq i$ and $X_i'$ does collide with some $X_k$ for $k \neq i$. It holds that $|\{X_1, \cdots, X_{i-1}, X_{i+1}, \cdots X_m\}| \leq m$. Therefore, the probability that a $X_i'$ lies in this set is $\leq m/n$. Since $\ell$ and $\ell_i'$ are symmetric, we have that also $\ell - \ell_i' = -1$ with probability $\leq m/n$.  Hence, $E[(\ell - \ell_i')^2] = P(|\ell - \ell_i'| = 1) \leq 2m/n$. Applying Efron-Stein we get
\[
Var(\ell) \leq \frac{1}{2}\cdot \sum_{i=1}^m E\ld[\ld(\ell - \ell_i'\rd)^2\rd] \leq \frac{m^2}{n}
\]
Now we just plug this bound on the variance into Chebyshev inequality and we get the desired inequality.
\end{proof}

\paragraph{Fingerprints.} Given a sample $S$, we define its \emph{fingerprint} $F$ as the set of tuples $(c_a, w(a))$ where for each distinct item $a$ in $S$, we add to $F$ such a tuple with $c_a$ being equal to the number of copies of $a$ in $S$.
Having a fingerprint of $S$ is sufficient for any algorithm, oblivious of $(U, w)$, to produce a sample $S'$ that is equal to $S$, up to relabeling of the elements.
Since the only allowed queries are testing equality of two items and the weight query, one may easily prove that the execution of the algorithm on these two samples is the same (indeed, these two samples are indistinguishable by the equality and weight queries).
Therefore, we can safely assume that an algorithm in the proportional setting takes as an input the fingerprint $F$ of $S$, rather than $S$. For algorithms in the hybrid setting, we can assume that it takes as input separately the fingerprint of the proportional and the fingerprint of the uniform samples.

\begin{lemma} \label{lem:set_size_lower_bound}
Let us have parameters $n$ and $\epsilon < 1/3$. Let $N = n$ with probability $1/2$, and $N = (1-\varepsilon)n$ otherwise. Consider the random instance of the sum estimation problem $(U, w)$ with $|U|=N$ and $w(a)=1$ for each $a\in U$. Consider a sample of size $m$ and its fingerprint $F_m = \{(c_i, w(a_i)\}_{i = 1 \dots \ell}$ and define the ratio 
\[
\Rr(F_m) = \frac{P\ld(F_m \,|\, N=n \rd)}{P\ld(F_m \,|\, N = (1-\varepsilon) n \rd)}.
\]
If $m = o(\sqrt{n} /\varepsilon)$ and $m=o(n)$, then 
\[
P\ld(\frac{98}{100} \leq \Rr(F_m) \leq \frac{100}{98}\rd) \geq \frac{99}{100}
\]
for $n$ large enough.
\end{lemma}
\begin{proof}
We can explicitly compute the likelihood of a given fingerprint $F_m = \{(c_i, w(a_i)\}_{i = 1 \dots \ell}$ where $\ell$ is the number of distinct elements as
\begin{align}
P\left(F_m \big| N = r \right) =& \frac{\ell!}{r^m} \binom{r}{\ell} \binom{m}{c_1 \dots c_\ell} \\
=& \frac{1}{r^{m-\ell}} \prod_{i=1}^{\ell-1} \ld(1 - \frac{i}{r}\rd) \binom{m}{c_1 \dots c_\ell}
\end{align}
and therefore
\begin{align}
\Rr(F_m) &= (1-\varepsilon)^{m-\ell} \cdot \prod_{i=1}^{\ell-1} \frac{1 - i/n}{1 - i/(1-\varepsilon)n} \\
&= (1-\varepsilon)^{m-\ell} \cdot \prod_{i=1}^{\ell-1} \ld( 1 + \frac{\varepsilon i}{(1-\varepsilon)n - i}\rd)
\end{align}
Note that $\Rr(F_m)$ depends only on $\ell$. From now on we denote it with $\Rr(\ell)$. 

Now we define an interval $[a, b]$ such that $P(\ell \in [a, b]) \geq 99/100$. To do so, we first compute the expectation of $\ell$ and then use the concentration bound of Lemma~\ref{lem:occupationCheby}. We prove that
\[
E\ld[\ell | N=n\rd] \leq E\ld[\ell | N=(1-\varepsilon)n\rd] \leq E\ld[\ell | N=n\rd] + O\ld(\varepsilon\frac{m^2}{n}\rd).
\]
The expression of $E\ld[\ell | N=n\rd]$ is given by
\[
E\ld[\ell | N=n\rd] = n\cdot \ld(1 - \ld(1-\frac{1}{n}\rd)^m\rd)
\]
since each item is not sampled with probability $(1-\frac{1}{n})^m$. From this expression, it is apparent that $E\ld[\ell | N=n\rd]$ is increasing in $n$. Expanding this formula, we get  
\begin{align}
E\ld[\ell | N=(1-\varepsilon)n\rd] - E\ld[\ell | N=n\rd] = \varepsilon\frac{m^2}{2n} + O\ld(\varepsilon\frac{m^3}{n^2}\rd) = O\ld(\varepsilon\frac{m^2}{n}\rd)
\end{align}
where the last estimate uses the $m = o(n)$ assumption.
Now we define $a = E\ld[\ell | N=(1-\varepsilon)n\rd] - 10 m/\sqrt{n}$ and $b = E\ld[\ell | N=n\rd] + 10 m/\sqrt{n}$. 
Using the last result we proved, we have
\[
|b - a| = 20\frac{m}{\sqrt{n}} + O\ld( \varepsilon\frac{m^2}{n}\rd) = o\ld(\frac{1}{\varepsilon}\rd).
\]
Note that, like all asymptotics in this proof, the $o(1/\epsilon)$ is for the limit $n \rightarrow +\infty$ and makes sense even for $\epsilon$ being a constant. Thanks to Lemma~\ref{lem:occupationCheby} we have $P(\ell \not\in [a, b] \,|\, N=n) \leq 1/100$ and $P(\ell \not\in [a, b] \,|\, N=(1-\varepsilon) n) \leq 1/100$, and therefore $P(\ell \not\in [a, b]) \leq 1/100$.

Now we give bounds on $\Rr(a)$ and $\Rr(b)$. It is apparent from the explicit formula above that $\ell \mapsto \Rr(\ell)$ is an increasing function. We have
\begin{align}
\Rr(a) \cdot \frac{99}{100}  &\leq \Rr(a)  \sum_{k \in [a, b] \cap \mathbb{Z}} P\ld(\ell = k \,|\, N = (1-\varepsilon)n\rd) \\
&\leq \sum_{k \in [a, b] \cap \mathbb{Z}} \Rr(k) P\ld(\ell = k \,|\, N = (1-\varepsilon)n\rd) \\
&\leq \sum_{k \in [a, b] \cap \mathbb{Z}} P\ld(\ell = k \,|\, N = n\rd) \leq 1
\end{align}
and thus, $\Rr(a) \leq 100/99$. Analogously,
\begin{align}
\frac{1}{\Rr(b)} \cdot \frac{99}{100}  &\leq   \frac{1}{\Rr(b)} \sum_{k \in [a, b] \cap \mathbb{Z}} P\ld(\ell = k \,|\, N = n\rd) \\
&\leq \sum_{k \in [a, b] \cap \mathbb{Z}} \frac{1}{\Rr(k)} P\ld(\ell = k \,|\, N = n\rd) \\
&\leq \sum_{k \in [a, b] \cap \mathbb{Z}} P\ld(\ell = k \,|\, N = (1-\varepsilon)n\rd) \leq 1
\end{align}
and thus, $\Rr(b) \geq 99/100$. We now have an upper bound on $\Rr(a)$ and a lower bound on $\Rr(b)$. However, we need a lower bound on $\Rr(a)$ and an upper bound on $\Rr(b)$ (that is, the other way around). For each $k < m$ we have
\[
\frac{\Rr(k+1)}{\Rr(k)} = \frac{1}{1-\varepsilon} \cdot \ld( 1 + \frac{\varepsilon k}{(1-\varepsilon)n - k}\rd) \leq (1+ 2\varepsilon) \cdot \ld( 1 + \frac{2\varepsilon m}{n}\rd) \leq 1 + 3\varepsilon
\]
for $n$ large enough, where we used $k<m=o(n)$ and $\varepsilon \leq 1/3$. Hence,
\begin{align}
\Rr(b) &\leq  \Rr(a) \cdot \ld(1+ 3\varepsilon \rd)^{\lceil |b-a|\rceil}  \\
&\leq \frac{100}{99} \cdot e^{3\varepsilon \lceil |b-a|\rceil}  \\
&\leq \frac{100}{99} \cdot e^{o(1)} \leq \frac{100}{98}
\end{align}
where the last inequality holds for $n$ large enough because $m=o(\sqrt{n}/\varepsilon)$. 
\begin{align}
\Rr(a) &\geq  \Rr(b) \cdot \ld(1+ 3\varepsilon \rd)^{-\lceil |b-a|\rceil}  \\
&\geq \frac{99}{100} \cdot e^{-3\varepsilon \lceil |b-a|\rceil}  \\
&\geq \frac{99}{100} \cdot e^{-o(1)} \geq \frac{98}{100}
\end{align}
Finally, we have for $n$ large enough that
\[
P\ld(\Rr(\ell) \not\in  \ld[\frac{98}{100},\frac{100}{98}\rd]\rd) \leq P\ld(\ell \not\in [a, b]\rd) \leq \frac{1}{100}.
\]
\end{proof}

Using the same approach as Lemma~\ref{lem:set_size_lower_bound}, we prove a similar result for the task of estimating the bias $p$ of a Bernoulli random variable up to an additive $\varepsilon$. In this setting, we provide an asymptotic lower bound on the number of samples, where the asymptotics are meant for the limit $(p,\varepsilon) \rightarrow 0$.  

\begin{lemma} \label{lem:biansed_coin_discrimination}
Let $0< \varepsilon < p$ and set $q=p$ with probability $1/2$, and $q=p-\varepsilon$ otherwise. Let $X_1 \dots X_m$ be a sequence of i.i.d. Bernoulli random variables with bias $q$, and let $\ell = \sum_{i=1}^m X_i \sim Bin(m, q)$.
Define the ratio
\[
\Rr(\ell) = \frac{P\ld(\ell \,|\, q=p\rd)}{P\ld(\ell \,|\, q=p-\varepsilon\rd)}.
\]
If $m=o(p/\varepsilon^2)$ then
\[
P\ld(\frac{98}{100} \leq \Rr(\ell) \leq \frac{100}{98} \rd) \geq \frac{99}{100}.
\]
for $p$ (and thus also $\epsilon)$ small enough.
\end{lemma}
\begin{proof}
We follow the same scheme we adopted in the proof of Lemma~\ref{lem:set_size_lower_bound}.
First, we compute $\Rr(\ell)$ explicitly
\begin{align}
\Rr(\ell) = \frac{p^\ell (1-p)^{m-\ell}}{(p-\varepsilon)^\ell (1-p+\varepsilon)^{m-\ell}} 
= \frac{\ld(1 + \frac{\varepsilon}{1-p}\rd)^{\ell -m}}{\ld(1 - \frac{\varepsilon}{p}\rd)^\ell} .
\end{align}
We have $E[\ell\,|\, q=p] = mp$, $E[\ell\,|\,q=p-\varepsilon] = m(p-\varepsilon)$, $Var(\ell\,|\,q=p) \leq mp$, and $Var(\ell\,|\,q=p-\varepsilon)\leq mp$. 
We define $a = m(p-\varepsilon) - 10\sqrt{mp}$ and $b=mp + 10\sqrt{mp}$, and using Chebyshev inequality we have $P(\ell \in [a, b] \,|\, q=p) \geq 99/100$ and $P(\ell \in [a, b] \,|\, q=p-\varepsilon) \geq 99/100$. Hence, $P(\ell \in [a, b]) \geq 99/100$. Notice that $|b-a| = m\varepsilon + 20 \sqrt{mp} = o(p/\varepsilon)$ where the second equality holds by the assumption that $m = o(p/\epsilon^2)$.
Now, we bound $\Rr(a)$ and $\Rr(b)$. Again, we notice that $\ell \mapsto \Rr(\ell)$ is an increasing function.
\begin{align}
\Rr(a) \cdot \frac{99}{100}  &\leq \Rr(a)  \sum_{k \in [a, b] \cap \mathbb{Z}} P\ld(\ell = k \,|\, q=p-\varepsilon\rd) \\
&\leq \sum_{k \in [a, b] \cap \mathbb{Z}} \Rr(k) P\ld(\ell = k \,|\, q=p-\varepsilon\rd) \\
&\leq \sum_{k \in [a, b] \cap \mathbb{Z}} P\ld(\ell = k \,|\, q=p\rd) \leq 1
\end{align}
and thus, $\Rr(a) \leq 100/99$. Analogously, we prove $\Rr(b) \geq 99/100$.
For each $k < m$ we have
\[
\frac{\Rr(k+1)}{\Rr(k)} = \frac{1 + \frac{\varepsilon}{1-p}}{1-\frac{\varepsilon}{p}} \leq 1 + 3\frac{\varepsilon}{p}
\]
for $p$ and $\varepsilon$ sufficiently small. Hence,
\begin{align}
\Rr(b) &\leq  \Rr(a) \cdot \ld(1+ 3\frac{p}{\varepsilon} \rd)^{\lceil |b-a|\rceil}  \\
&\leq \frac{100}{99} \cdot e^{3\frac{p}{\varepsilon} \lceil |b-a|\rceil}  \\
&\leq \frac{100}{99} \cdot e^{o(1)} \leq \frac{100}{98}
\end{align}
where the last inequality holds for $p$ and $\epsilon$ sufficiently small. Analogously, we prove $\Rr(a) \geq 98/100$. Finally, we have
\[
P\ld(\Rr(\ell) \not\in  \ld[\frac{98}{100},\frac{100}{98}\rd]\rd) \leq P\ld(\ell \not\in [a, b]\rd) \leq \frac{1}{100}.
\]
\end{proof}

\subsection{Proportional sampling.}

In this section, we assume, as we did in Section~\ref{sec:weighted_sampling}, that we can sample only proportionally. We prove that $\Omega(\sqrt{n}/\varepsilon)$ samples are necessary to estimate $W$ with probability $2/3$, thus $\PropEstimator$ is optimal up to a constant factor. 

\paragraph{Deciding the number of samples at run-time.}
In all our lower bounds, we show that it is not possible that an algorithm takes $m=o(T(n, \varepsilon))$ samples in the \emph{worst case} and correctly approximates $W$ with probability $2/3$.
All these lower bounds safely extend to lower bounds on the expected number of samples $E[m]$. All our proofs work by showing that the Bayes classifier has risk $1/2-o(1)$. Suppose now that we have an algorithm $A$ that uses in expectation $\mu(n,\epsilon) = o(T(n, \epsilon))$ samples. We now define a classifier as follows. We run $A$ and abort if it uses more than $20\mu(n,\epsilon)$ samples\footnote{Note that, while the algorithm does not know $\mu(n,\epsilon)$, this is not an issue in this argument. The reason is that a classifier is defined as an arbitrary function from the set of possible samples and private randomness to the set of classes. This allows us to ``embed" $\mu$ into the classifier}. We return the answer given by $A$ or an arbitrary value if we have aborted the algorithm. By the Markov's inequality, the probability that we abort is at most $1/20$. Our classifier has risk $1/3 + 1/20 < 2/5$. Since any constant success probability greater than $1/2$ is equivalent up to probability amplification, we also have a classifier with risk $1/3$ that uses $O(\mu(n,\epsilon)) = o(T(n,\epsilon)$ samples. Since a Bayes classifier with such parameters does not exist (as we show) and Bayes classifiers are risk-optimal, this is a contradiction.


\begin{theorem} \label{thm:lb_prop}
In the proportional setting, there does not exist an algorithm that, for every instance $(U, w)$ with $|U|=n$, takes $m$ samples for $m=o(\sqrt{n}/\varepsilon)$ and returns a $\onepmepsilon$-approximation of $W$ with probability $2/3$. This holds also when $n$ is known to the algorithm.
\end{theorem}
\begin{proof}
As already proven, we may assume that the algorithm only gets the fingerprint $F_m$ of the sample $S$ of size $m$, instead of $S$ itself. 
In the rest of the proof, we separately consider two cases: $\varepsilon \geq 1/\sqrt{n}$ and $\varepsilon < 1/\sqrt{n}$.

\paragraph{Case $\varepsilon \geq \frac{1}{\sqrt{n}}$:} We first define the hard instance $(U, w)$. Define the random variable $k$ as $k=(1-\varepsilon) n$ with probability $1/2$ and $k=n$ otherwise, then let $U=\{a_1 \dots a_n\}$ and 
\[
  w(a_i) = 
  \begin{cases}
    1 & \text{if } i \leq k \\
    0 & \text{otherwise.}
  \end{cases}
\]
Items with weight zero are never sampled while sampling proportionally while we are sampling uniformly from those with weight 1. Moreover $m=o(\sqrt{n}/\varepsilon)=o(n)$ for $\varepsilon \geq 1/\sqrt{n}$. Hence, this is exactly the settings of Lemma~\ref{lem:set_size_lower_bound}. If we let $F_m$ to be the fingerprint of the samples and define $\Rr(F_m)$ as in \Cref{lem:set_size_lower_bound}, we have 
\[
P\ld(\frac{98}{100} \leq \Rr(F_m) \leq \frac{100}{98}\rd) \geq \frac{99}{100}.
\]
Applying Lemma~\ref{lem:bayesclas} gives us the desired result. 

\paragraph{Case $\varepsilon < \frac{1}{\sqrt{n}}$:}
For convenience, we show an instance of size $n+1$ instead of $n$. First, we define $s^2= \min\big\{\frac{\sqrt{n}}{\varepsilon m}, \frac{n}{4}\big\}$
and notice that $s = \omega(1)$ and $s \leq \sqrt{n} / 2$.
Define the random variable $k$ as  $k = n - s \sqrt{n}$ with probability $1/2$ and $k = n$ otherwise. Define the events $\Ee_1 = \{k = n - s \sqrt{n}\}$ and $\Ee_2=\{k=n\}$. We construct $(U, w)$ so that $U= \{a_0 \dots a_n\}$ and 
\[
  w(a_i) = 
  \begin{cases}
    \frac{s\sqrt{n}}{\varepsilon} - n & \text{if } i = 0 \\
    1 & \text{if } 1 \leq i \leq k \\
    0 & \text{otherwise.}
  \end{cases}
\]
 Notice that our choice of $s$ together with $\varepsilon < 1/\sqrt{n}$ guarantees $n - s\sqrt{n} \geq n/2$ and $s\sqrt{n}/\varepsilon - n = \omega(n)$. We have that $W = \sum_{i=0}^n w(a_i)$ differs by more than a factor of $(1 +\varepsilon)$ between events $\Ee_1$ and $\Ee_2$. 

Consider an element $a\in U$ sampled proportionally and define $p_i = P(a \neq a_0 \,|\, \Ee_i)$ for $i=1, 2$. Then,
\begin{align}
p_2 &= \frac{n}{n-s\sqrt{n} + w(a_0)} = \frac{\varepsilon \sqrt{n}}{s} \\
p_1 &= \frac{n -s\sqrt{n}}{n-s\sqrt{n} + w(a_0)} = \frac{p_2 -\varepsilon}{1-\varepsilon} 
\end{align}
and $|p_1 -p_2| \leq \varepsilon$.
We perform $m$ samples in total and define the random variable $X$ counting the number of times items other than $a_0$ are sampled. We define $F_X$ as the fingerprint of the sampled items different from $a_0$. 
Given $F_X$, we may easily reconstruct the fingerprint of the whole sample by adding the tuple $(m - X, w(a_0))$. It thus holds $P(F_m | \Ee_i) = P(F_X | \Ee_i)$ for $i \in \{1,2\}$. 

Now we define the event $\Ll = \{ X \leq 30 E[X] \}$ and by Markov's inequality $P(\Ll) \geq 29/30$. Moreover, 
\[
E[X] = m (p_1 + p_2)/2 \leq m (p_2 + \varepsilon) \leq m (\varepsilon\sqrt{n}/s + \varepsilon) \leq 2\varepsilon m \sqrt{n}/s = o\ld(\frac{n}{s}\rd).
\]
Define $\varepsilon' = s/\sqrt{n}$. Conditioning on $\Ll$, we have $X \leq 30 E[X] =o(\sqrt{n} / \varepsilon') = o(n)$. If we further condition on $X=x$ for some $x \leq 30E[X]$ and we look at $F_X$, we are exactly in the setting of Lemma~\ref{lem:set_size_lower_bound}. Therefore,
\[
P\ld(\frac{98}{100} \leq \frac{P\ld(F_X\,|\, X=x, \Ee_1\rd)}{P\ld(F_X\,|\, X=x, \Ee_2\rd)} \leq \frac{100}{98}\rd) \geq \frac{99}{100}
\]
and integrating over $\Ll$ we obtain
\[
P\ld(\frac{98}{100} \leq \frac{P\ld(F_X\,|\, X, \Ee_1\rd)}{P\ld(F_X\,|\, X, \Ee_2\rd)} \leq \frac{100}{98} \,\bigg|\, \Ll\rd) \geq \frac{99}{100}.
\]
Now, we will bound the ratio 
\[
\Rr(X) = \frac{P\ld(X\,|\,\Ee_1\rd)}{P\ld(X\,|\,\Ee_2\rd)}.
\]
We have $X = \sum_{i=1}^m X_i$, where $X_i$ is an indicator for the $i$-th sample not being equal to $a_0$. Therefore, $X\,|\,\Ee_j \sim Bin(m, p_j)$ for $j=1,2$.
It holds, $|p_1-p_2| \leq \varepsilon$. Because $m = o(\sqrt{n}/\epsilon)$ and by the definition of $s$, we have $m= o(\sqrt{n}/(s\varepsilon)) = o(p_2/\varepsilon^2)$. We can apply \Cref{lem:biansed_coin_discrimination} and obtain 
\[
P\ld(\frac{98}{100} \leq \Rr(X) \leq \frac{100}{98} \rd) \geq \frac{99}{100}.
\]
Finally, we put the bounds together. We consider the ratio
\[
\Rr(F_m) = \frac{P\ld(F_m \,|\, \Ee_1\rd)}{P\ld(F_m \,|\, \Ee_2\rd)} = \frac{P\ld(F_X \,|\, \Ee_1\rd)}{P\ld(F_X \,|\, \Ee_2\rd)} = 
\frac{P\ld(X\,|\, \Ee_1\rd) \cdot P\ld(F_X \,|\, X, \Ee_1\rd)}{P\ld(X \,|\, \Ee_2\rd) \cdot P\ld(F_X \,|\, X, \Ee_2\rd)} .
\]
By the union bound, along with $7/8 < (98/100)^2$, we get
\begin{align}
P\ld(\Rr(F_m) \not\in \ld[\frac{7}{8},\frac{8}{7}\rd]\rd) &\leq \\
P\ld(\Bar{\Ll}\rd) + P\ld(\frac{P\ld(F_X\,|\, X, \Ee_1\rd)}{P\ld(F_X\,|\, X, \Ee_2\rd)} \not\in \ld[\frac{98}{100},\frac{100}{98}\rd]\,\bigg|\, \Ll\rd) + P\ld(\frac{P\ld(X\,|\, \Ee_1\rd)}{P\ld(X\,|\, \Ee_2\rd)} \not\in \ld[\frac{98}{100},\frac{100}{98}\rd]\rd) &\leq \\
\frac{1}{30} + \frac{1}{100} + \frac{1}{100} &\leq \frac{1}{15}.
\end{align}
We can apply \Cref{lem:bayesclas} and conclude the proof.
\end{proof}

\subsection{Sum estimation in hybrid setting, known \texorpdfstring{$n$}{n}.}
In this section, we assume, as we did in Section~\ref{sec:hybrid_sampling}, that we can sample both proportionally and uniformly. We will prove that $\Omega(\min(\sqrt[3]{n}/\varepsilon^{4/3}, n))$ samples are necessary to estimate $W$ with probability $2/3$. This complexity is the minimum of the sample complexity of the $\HybridEstimator$ and (up to a logarithmic factor) the complexity of the standard coupon collector algorithm.

\begin{theorem} \label{thm:lb_hybrid}
In the hybrid setting, there does not exist an algorithm that, for every instance $(U, w)$ with $|U|=n$, takes $m=o(\min(\sqrt[3]{n}/\varepsilon^{4/3}, n))$ proportional and uniform samples and returns a $\onepmepsilon$-approximation of $W$ with probability $2/3$. This holds also when $n$ is known to the algorithm.
\end{theorem}

\begin{proof}
It is enough to prove that for $\varepsilon \geq 8/\sqrt{n}$, any algorithm returning a $\onepmepsilon$-approximation of $W$ with probability $2/3$ must take $\Omega(\sqrt[3]{n}/\varepsilon^{4/3})$ samples. Indeed, if $\varepsilon < 8/\sqrt{n}$, a $(1\pm \varepsilon)$-approximation is also a $(1\pm 8/\sqrt{n})$-approximation, and then $\Omega(n)$ samples are necessary. In either case we then need $\Omega(\min(\sqrt[3]{n}/\varepsilon^{4/3}, n))$ samples.

Define the random variable $k$ as  $k = (1-\varepsilon) n^{2/3}/\varepsilon^{2/3}$ with probability $1/2$ and $k = n^{2/3}/\varepsilon^{2/3}$ otherwise. Define the events $\Ee_1 = \{k = (1-\varepsilon) n^{2/3}/\varepsilon^{2/3}\}$ and $\Ee_2=\{k=n^{2/3}/\varepsilon^{2/3}\}$. We construct $(U, w)$ so that $U= \{a_1 \dots a_n\}$ and 
\[
  w(a_i) = 
  \begin{cases}
    1 & \text{if } 1 \leq i \leq k \\
    0 & \text{otherwise.}
  \end{cases}
\]
We have that $W = \sum_{i=1}^n w(a_i)$ differs by more than a factor of $(1 +\varepsilon)$ between events $\Ee_1$ and $\Ee_2$. Notice that $k \leq n/4$ as $\varepsilon \geq 8/\sqrt{n}$.
Let $S_p$ be the multiset of proportional samples and $S_u$ the multiset of uniform samples. Let $T_h = (S_p \cup S_u) \cap \{k+1, \cdots, n\}$ and $T_l = (S_p \cup S_u) \cap \{1, \cdots, k\}$. Let $F_l=\{(c_i, w(a_i))\}_{i}$ be the fingerprint of $T_l$ and $F_h=\{(c_i, w(a_i))\}_{i}$ of be the fingerprint of $T_h$. 
We now argue hat $(F_l,F_h)$ is sufficient for any algorithm, oblivious of the choice of $k$, to reconstruct sample multisets $S_u'$ and $S_p'$ distributed as $S_u$ and $S_p$. 
We pick $|F_l| - m$ items at random from $F_l$ and let $S_u'$ be the multiset of these items, together with all items in $F_h$. We let $S_p'$ be the multiset of the items left in $F_l$. It is easy to verify that $(S_u',S_p') \sim (S_u, S_p)$. 
Thus, we can assume that the algorithm is given $(F_l, F_h)$ as input, instead of the sample multisets $S_u$ and $S_p$. 

Consider $a \in U$ sampled uniformly, and define $p_i = P(w(a) = 1  \,|\, \Ee_i)$ for $i=1, 2$. Then,
\begin{align}
p_1 &= \frac{n^{2/3}/\varepsilon^{2/3}}{n} = \frac{1}{n^{1/3}\varepsilon^{2/3}} \\
p_2 &= \frac{(1-\varepsilon)n^{2/3}/\varepsilon^{2/3}}{n} = \frac{1}{n^{1/3} \varepsilon^{2/3}} - \frac{\varepsilon^{1/3}}{n^{1/3}}
\end{align}
and defining $\varepsilon' =\varepsilon^{1/3}/n^{1/3} $ we obtain $p_2 = p_1 - \varepsilon'$. Let $X = |S_u \cap \{1, \cdots, k\}|$. Since $X\,|\,\Ee_j \sim Bin(m, p_j)$ for $j=1, 2$ and $m=o(n^{1/3}/\varepsilon^{4/3}) = o(p_1 / \varepsilon'^2)$, we can apply \Cref{lem:biansed_coin_discrimination} and obtain
\[
P\ld(\frac{98}{100} \leq \frac{P\ld(X\,|\,\Ee_1\rd)}{P\ld(X\,|\,\Ee_2\rd)} \leq \frac{100}{98} \rd) \geq \frac{99}{100}.
\]
Conditioning on $X=x$ for some $x=1 \dots m$, $F_l$ represents a sample of $x + |S_p|$ items drawn uniformly from a set of cardinality $k$, so we are in the setting of \cref{lem:set_size_lower_bound}. Moreover, we have 
\[
|F_l| \leq |S_p| + |S_u| = o\ld(\frac{\sqrt[3]{n}}{\varepsilon^{4/3}}\rd) = o\ld(\frac{\sqrt{n^{2/3}/\varepsilon^{2/3}}}{\varepsilon}\rd).
\]
Hence, \Cref{lem:set_size_lower_bound} holds and we have
\[
P\ld(\frac{98}{100} \leq \frac{P\ld(F_l\,|\, X=x, \Ee_1\rd)}{P\ld(F_l\,|\, X=x, \Ee_2\rd)} \leq \frac{100}{98}\rd) \geq \frac{99}{100}
\]
and integrating over $x = 1 \dots m$ we have
\[
P\ld(\frac{98}{100} \leq \frac{P\ld(F_l\,|\, X, \Ee_1\rd)}{P\ld(F_l\,|\, X, \Ee_2\rd)} \leq \frac{100}{98} \rd) \geq \frac{99}{100}.
\]

Similarly, we have that $|F_h| \leq |S_u| = o(\sqrt[3]{n}/\varepsilon^{4/3}) = o(\sqrt{n}/\varepsilon)$ where the second inequality holds because we are assuming $\varepsilon > 8/\sqrt{n}$. Moreover, conditioning on $X=x$ for some $x=1 \dots m$, $F_h$ represent a sample of $|S_u| - x$ items drawn uniformly from a set of size $n -k$. It holds  $n-k\geq 3n/4$, and $n-k$ thus differs by at most a factor $1-\varepsilon$ between the two events $\Ee_1,\Ee_2$.  Again, we are in the right setting to apply \Cref{lem:set_size_lower_bound}, and integrating over $x=1 \dots m$ gives
\[
P\ld(\frac{98}{100} \leq \frac{P\ld(F_h\,|\, X, \Ee_1\rd)}{P\ld(F_h\,|\, X, \Ee_2\rd)} \leq \frac{100}{98}\rd) \geq \frac{99}{100}.
\]
We are now ready to put everything together. Note that $F_l$ and $F_h$ are independent once conditioned on $(\Ee_1,X)$ or $(\Ee_2,X)$. Define the ratio
\[
\Rr(F_l, F_h) = \frac{P\ld((F_l, F_h) \,|\, \Ee_1\rd)}{P\ld((F_l, F_h) \,|\, \Ee_2\rd)} = 
\frac{P\ld(X\,|\, \Ee_1\rd) \cdot P\ld(F_l \,|\, X, \Ee_1\rd) \cdot P\ld(F_h \,|\, X, \Ee_1\rd)}{P\ld(X\,|\, \Ee_2\rd) \cdot P\ld(F_l \,|\, X, \Ee_2\rd) \cdot P\ld(F_h \,|\, X, \Ee_2\rd)}
\]
Using the union bound, along with $7/8 < (98/100)^3$, we get
\begin{align}
P\ld(\Rr(F_l, F_h) \not\in \ld[\frac{7}{8},\frac{8}{7}\rd]\rd) &\leq \\
 P\ld(\frac{P\ld(X\,|\,\Ee_1\rd)}{P\ld(X\,|\, \Ee_2\rd)}\not\in \ld[\frac{98}{100},\frac{100}{98}\rd]\rd) + P\ld(\frac{P\ld(F_l\,|\, X, \Ee_1\rd)}{P\ld(F_l\,|\, X, \Ee_2\rd)} \not\in \ld[\frac{98}{100},\frac{100}{98}\rd]\rd) +&\\ P\ld(\frac{P\ld(F_h\,|\, X, \Ee_1\rd)}{P\ld(F_h\,|\, X, \Ee_2\rd)} \not\in \ld[\frac{98}{100},\frac{100}{98}\rd]\rd) &\leq \\
\frac{1}{100} + \frac{1}{100} + \frac{1}{100} &< \frac{1}{15}.
\end{align}
We can apply \Cref{lem:bayesclas} and conclude the proof.

\end{proof}

\subsection{Sum estimation in hybrid setting, unknown \texorpdfstring{$n$}{n}.}
We now prove a lower bound for the hybrid setting, in case the algorithm does not know $n$.

\begin{theorem}
In the hybrid setting, there does not exist an algorithm that, for every instance $(U, w)$, takes $m=o(\min(\sqrt{n}/\varepsilon, n))$ samples and returns a $\onepmepsilon$-approximation of $W$ with probability $2/3$. This holds also when the algorithm is provided with an advice $\tilde{n}$ such that  $(1-\varepsilon) n \leq \tilde{n} \leq n$.
\end{theorem}
\begin{proof}
Employing the same argument as in \Cref{thm:lb_hybrid}, it is sufficient to prove that for $\varepsilon \geq 1/\sqrt{n}$ a lower bound of $\Omega(\sqrt{n}/\varepsilon)$ holds.

Consider the instance $(U, w)$ where $w(a)=1$ for each $a\in U$ and we set $|U| = n$ with probability $1/2$ and $|U|=(1-\varepsilon)n$ otherwise. Providing a $\onepmepsilon$-approximation of $W$ is equivalent to distinguishing between the two cases. On this instance, sampling uniformly and proportionally is the same. Therefore, we are in the setting of \Cref{lem:set_size_lower_bound}. Combining \Cref{lem:set_size_lower_bound} and \Cref{lem:bayesclas} like in the proofs above, we get that no classifier can distinguish between $|U|=n$ and $|U|=(1-\varepsilon)n$ with probability $2/3$ using $o(\sqrt{n}/\varepsilon)$ samples.
\end{proof}

\section{Counting Edges in a Graph.} \label{sec:approximatelycountingedges}

In this section, we show an algorithm that estimates the average degree of a graph $G=(V, E)$ in the model in which we are allowed to perform random vertex queries, random edge queries, and degree queries. Recall that a \emph{random vertex query} returns a uniform sample form $V$, a \emph{random edge queries} returns a uniform sample from $E$, and a \emph{degree queries} returns $deg(v)$ when we provide $v \in V$ as argument.
In this section, we denote the number of vertices and edges with $n$ and $m$ respectively.

Here, we show that $\harmonicestimator$ from \Cref{sec:hybrid_sampling}, can be adapted to estimate the average degree
$d$. 
It achieves a complexity of $O(\frac{m \log \log n}{n' \varepsilon^2} + \frac{n}{ n'\varepsilon^2})$ in expectation, where $n'$ is the number of non-isolated\footnote{Recall that a vertex is isolated if it has degree $0$.} vertices.
This is very efficient when there are few isolated vertices and the graph is sparse. 
Moreover, the only way we use sampling of vertices is to estimate the number of non-isolated vertices. Therefore, if we assume that there are no isolated vertices in the graph, it is sufficient to only be able to uniformly sample edges.

Our approach is similar to that of \cite{Katzir2011} but the authors in the paper do not prove bounds on the time complexity. Moreover, their algorithm only works when there are no isolated vertices.

\begin{theorem}
Given a graph $G=(V, E)$, consider a model that allows (1) random vertex queries, (2) random edge queries, and (3) degree queries. 
In this model, there exists an algorithm that, with probability at least $2/3$, returns a $ \onepmepsilon$-approximation $\hat{d}$ of the average degree $d=2m/n$. This algorithm performs $O(\frac{m \log \log d}{n' \varepsilon^2} + \frac{n}{ n'\varepsilon^2})$ queries in expectation, where $n'$ is the number of non-isolated vertices.
\end{theorem}
\begin{proof}
We first show an algorithm that is given $\tilde{\theta}$ such that $d \leq \tilde{\theta}$, has time complexity $O(\frac{\tilde{\theta}n}{\varepsilon^2 n'} +\frac{n}{\varepsilon^2 n'})$, and is correct with probability $2/3$. We define a sum estimation problem $(U,w)$. We set the universe to be $U = V$ and for each vertex $v \in U$, we set its weight $w(v) = deg(v)$. Then $W = \sum_{a\in U} w(a)=2m$ and $W/n = d$.
Sampling an edge uniformly at random and picking one of its endpoints at random corresponds to sampling a vertex proportionally to its weight. Moreover, we can sample vertices uniformly. Therefore, we are able to implement both queries of the hybrid setting.
We run $\harmonicestimator(\varepsilon, \tilde{\theta},1)$. By \Cref{lem:estimate_by_harm_mean}, it returns with probability at least $2/3$ a $1\pm\varepsilon$-approximation of $d$, and its sample complexity is $O(\frac{\tilde{\theta}n}{\varepsilon^2 n'} +\frac{n}{\varepsilon^2 n'})$ since what is called $p$ in \Cref{lem:estimate_by_harm_mean} is now $n'/n$.

It remains to get rid of the need for advice $\tilde{\theta}$. We use the standard technique of performing a geometric search. See, for example, \cite{Goldreich2006} for more details. 
We initialize $\tilde{\theta} = 1$ and in each subsequent iteration, we double $\tilde{\theta}$.
Let $K$ be a large enough constant. In each iteration, we run $K \log \log \tilde{\theta}$ independent copies of  $\harmonicestimator(\varepsilon, \tilde{\theta}, 1)$ and denote with $d_1 \dots d_{K\log \log \tilde{\theta}}$ the returned estimates. We define $\hat{d}$ as the median of $d_1 \dots d_{K\log \log \tilde{\theta}}$. We say that the $d_i$ succeeds if both the following hold: (i) $d_i \geq d/20$, (ii) $\tilde{\theta} < d$ or $d_i$ is a $\onepmepsilon$-approximation of $d$. Otherwise we say that $d_i$ fails. We extend this definition to $\hat{d}$.
Thanks to \Cref{lem:estimate_by_harm_mean}, $d_i$ fails with probability $\leq 1/3 + 1/20$.
By a standard argument based on the Chernoff bound, for $K$ large enough, we have that $\hat{d}$ fails with probability at most $2/( \pi \log^2 (2\tilde{\theta}))$. Denote with $\Ee_j$ the event that $\hat{d}$ succeeds at iteration $j$ (i.e., when $\tilde{\theta} = 2^{j-1}$). Define $\Ee=\bigcap_{j\geq 0} \Ee_j$. 
Union bound gives $P(\Ee) \geq 1 - \frac{2}{\pi} \sum_{j > 0} \frac{1}{j^2} =  2/3$.
We stop our algorithm when $\hat{d} \leq \tilde{\theta}/20$ and return $\hat{d}$. Conditioning on $\{\hat{d} \leq \tilde{\theta}/20\} \cap \Ee$, we have $\tilde{\theta}/20 \geq \hat{d} \geq d/20$. This implies that $\tilde{\theta} \geq d$ and hence $\hat{d}$ is a $\onepmepsilon$-approximation of $d$. Since $P(\Ee) \geq 2/3$, we have that with probability $2/3$ the returned value is a $\onepmepsilon$-approximation of $d$. 

One iteration of our algorithm has time complexity $O(\frac{\tilde{\theta} n \log \log \tilde{\theta}}{\varepsilon^2 n'})$.
%
We argue that the expected complexity is dominated by the first iteration in which $\tilde{\theta} \geq 40\, d$. 
The time complexity of each additional iteration (conditioned on being executed) increases by a factor $2 + o(1)$. 
Each additional iteration is executed only if the previous one resulted in an estimate $\hat{d} > \tilde{\theta}/20 \geq 2\,d$. This happens only when $\hat{d}$ is not a $\onepmepsilon$-approximation of $d$, and assuming a correct advice $\tilde{\theta} \geq d$ this happens outside of $\Ee$. Therefore, we execute each additional iteration with probability $\leq 1/3$. Since the time spend in each iteration increases by a factor $2+o(1)$ while the probability of executing the iteration decreases by a factor of $3$, the expected complexity contributed by each additional iteration for $\tilde{\theta} \geq 40\, d$ decreases by a factor of $2/3+o(1)$. Therefore, the expected complexity is dominated (up to a constant factor) by the first execution with $\tilde{\theta} \geq 40\, d$. If $d \geq 1$, then in this iteration, we have $\tilde{\theta} = \Theta(d)$. The time complexity is then $O(\frac{m \log \log d}{\epsilon^2 n'})$. If $d < 1$, then it holds $\theta = O(1)$ in this execution. The complexity is then $O(\frac{n}{n' \epsilon^2})$. This gives total time complexity of $O(\frac{m \log \log d}{n' \epsilon^2} + \frac{n}{n' \epsilon^2})$.
\end{proof}

\section{Open Problems.} \label{sec:open_problems}
We believe there are many interesting open problems related to our work. We now give a non-comprehensive list of questions that we think would give more understanding of weighted sampling and its applications.

\paragraph{More efficient algorithm for spacial classes of inputs.} Are there some large classes of inputs for which it is possible to get a more efficient algorithm? Can the problem be parameterized by some additional parameters apart from $n,\varepsilon$ (e.g. empirical variance) that tend to be small in practice?

\paragraph{Different sampling probabilities.}
Are there settings where one may efficiently sample with probability depending on the value of an item but not exactly proportional? Could this be used to give a general algorithm for estimating the sum $W$? An example of such a result is \cite{Dasgupta2014} where the authors show an efficient algorithm for estimating the average degree of a graph when sampling vertices with probabilities proportional to $\frac{m}{n} + d(v)$.

\paragraph{Get a complete understanding of edge counting.} The complexity of the problem of approximately counting edges in a graph is understood in terms of $n,m$ in the setting where we can only sample vertices uniformly at random. What is the complexity of counting edges when we allow only random edge queries? What if both random edge and random vertex queries are allowed? As we show, it may be useful to parameterize the problem by the fraction of vertices that are not isolated. What is the complexity of the problem under such parameterization?


\section*{Acknoweldgements.}
We are grateful to Rasmus Pagh for his advice. We would like to thank our supervisor, Mikkel Thorup, for helpful discussions and his support. We would like to thank Talya Eden for her advice regarding related work.

\bibliographystyle{siam}
\bibliography{literature}
\end{document}